\documentclass[copyright,creativecommons]{eptcs}
\providecommand{\event}{TARK 2023} 

\usepackage{iftex}

\ifpdf
  \usepackage{underscore}         
  \usepackage[T1]{fontenc}        
\else
  \usepackage{breakurl}           
\fi

\title{Comparing the Update Expressivity of Communication Patterns and Action Models}
\author{
Armando Casta\~neda
\institute{Instituto de Matem\'aticas \\Universidad Nacional Aut\'onoma de M\'exico}
\email{armando.castaneda@im.unam.mx}
\and
Hans van Ditmarsch
\institute{University of Toulouse, CNRS, IRIT}
\email{hans.van-ditmarsch@irit.fr}
\and
David A.\ Rosenblueth
\institute{Instituto de Inv.\ en Matem\'aticas Aplicadas y en Sistemas \\
Universidad Nacional Aut\'onoma de M\'exico}
\email{drosenbl@unam.mx}
\and
Diego A.\ Vel\'azquez
\institute{Posgr.\ en Ciencia e Ingenier\'ia de la Computaci\'on \\
Universidad Nacional Aut\'onoma de M\'exico}
\email{velazquez-diego@ciencias.unam.mx}
}
\def\titlerunning{Comparing the Update Expressivity of Communication Patterns and Action Models}
\def\authorrunning{Casta\~neda et al.}

\usepackage{amsfonts, amsmath, amssymb}
\usepackage{color}
\usepackage{tikz,url}
\usetikzlibrary{calc,shapes}

\newcommand{\eq}{\leftrightarrow}
\newcommand{\Eq}{\Leftrightarrow}
\newcommand{\imp}{\rightarrow}
\newcommand{\Imp}{\Rightarrow}

\newcommand{\et}{\wedge}
\newcommand{\vel}{\vee}
\newcommand{\Et}{\bigwedge}

\renewcommand{\phi}{\varphi}
\newcommand{\union}{\cup}
\newcommand{\Union}{\bigcup}
\newcommand{\inter}{\cap}
\newcommand{\Inter}{\bigcap}

\newcommand{\power}{\mathcal P}

\newcommand{\M}{\widehat{K}}

\newcommand{\bisim}{{\raisebox{.3ex}[0mm][0mm]{\ensuremath{\medspace \underline{\! \leftrightarrow\!}\medspace}}}}
\newcommand{\Nat}{\mathbb N}
\newcommand{\Naturals}{\Nat}
\newcommand{\domain}{\mathcal{D}}

\usepackage{newproof}

\newtheorem{theorem}{Theorem}
\newtheorem{example}[theorem]{Example}
\newtheorem{definition}[theorem]{Definition}
\newtheorem{proposition}[theorem]{Proposition}

\newtheorem{corollary}[theorem]{Corollary}

\newtheorem{lemma}[theorem]{Lemma}

\newcommand{\lang}{\mathcal L}

\newcommand{\weg}[1]{}
\newcommand{\pre}{\mathsf{pre}}

\newcommand{\cg}{R}

\newcommand{\cp}{\pmb{R}}

\newcommand{\byz}{\pmb{\textit{Byz}}}
\newcommand{\IS}{\pmb{\textit{IS}}}
\newcommand{\Sq}{\pmb{\textit{Sq}}}

\newcommand{\U}{\pmb{U}}
\newcommand{\SSigma}{\pmb{\Sigma}}
\newcommand{\view}{\mathsf{view}}

\begin{document}
\maketitle

\begin{abstract}
Any kind of dynamics in dynamic epistemic logic can be represented as an action model. Right? Wrong! In this contribution we prove that the update expressivity of communication patterns is incomparable to that of action models. Action models, as update mechanisms, were proposed by Baltag, Moss, and Solecki in 1998 and have remained the nearly universally accepted update mechanism in dynamic epistemic logics since then. Alternatives, such as arrow updates that were proposed by Kooi and Renne in 2011, have update equivalent action models. More recently, the picture is shifting. Communication patterns are update mechanisms originally proposed in some form or other by {\AA}gotnes and Wang in 2017 (as resolving distributed knowledge), by Baltag and Smets in 2020 (as reading events), and by Vel\'azquez, Casta\~neda, and Rosenblueth in 2021 (as communication patterns). All these logics have the same expressivity as the base logic of distributed knowledge. However, their update expressivity, the relation between pointed epistemic models induced by such an update, was conjectured to be different from that of action model logic. Indeed, we show that action model logic and communication pattern logic are incomparable in update expressivity. We also show that, given a history-based semantics and when restricted to (static) interpreted systems, action model logic is (strictly) more update expressive than communication pattern logic. Our results are relevant for distributed computing wherein oblivious models involve arbitrary iteration of  communication patterns.
\end{abstract}

\section{Introduction}

It is well known that the expressivity of public announcement logic is the same as that of epistemic logic \cite{plaza:1989}. This is proved by way of a reduction system showing that every public announcement formula is equivalent to one without public announcement modalities. Similarly, the expressivity of the logic of distributed knowledge with public announcements is the same as that of the logic of distributed knowledge \cite{AgotnesW17}. Again, this is shown by a reduction. A reduction also exists for the logic of distributed knowledge with action models \cite{baltagetal:1998}; see \cite[Fig.\ 5 and Th.\ 15]{WangA15} and the reduction axiom called AD in \cite[Fig.\ 9]{WangA15}.

Distributed knowledge can also be extended with dynamic modalities for communication patterns, an update mechanism proposed in \cite{diego:2021}. The resulting communication pattern logic is as expressive as the logic of distributed knowledge: we can reduce formulas with dynamic modalities to formulas without \cite{cdrv:2022}. This logic is a slight generalization of logics with similar modalities also showing this by reduction \cite{AgotnesW17,Baltag20}. A detailed comparison to these other proposals is found in \cite{cdrv:2022}.

We conclude that the logic of communication patterns and distributed knowledge has the same expressivity as the logic of action models and distributed knowledge, because they both reduce to the logic of distributed knowledge. A different matter, however, is so-called \emph{update expressivity} \cite{jveetal:2012,kooirenne,hvdetal.aus:2020}.

We will compare the update expressivity of communication pattern logic and action model logic. Communication patterns, like action models, are (induce) {\em updates} transforming pointed epistemic models into other pointed epistemic models. Is there for each communication pattern an action model defining the same update, and vice versa? Communication patterns can always be executed, but action models cannot always be executed, for example a truthful public announcement of $p$ requires $p$ to be true in some world. We can therefore expect a trivial difference in update expressivity. It becomes non-trivial if we also consider union of relations, such as non-deterministic choice between the announcement of $p$ and the announcement of $\neg p$. 

This is an overview of the structure of our contribution. Sect.~\ref{sect:tp} recalls communication pattern logic, action model logic, and update expressivity. In Sect.~\ref{sect:iamis} we show that for each communication pattern there is an update equivalent action model when executed on epistemic models that are interpreted systems. However, the resulting model may not be an interpreted system. 
In Sect.~\ref{sect:inc} we then show that communication pattern logic and action model logic are indeed incomparable in update expressivity on the class of epistemic models.
Finally, in Sect.~\ref{sect:hist} we propose a history-based semantics for communication pattern logic for which the class of interpreted systems is, after all, closed under updates, and we show that for each iterated communication pattern there is then an update equivalent action model.

\section{Communication pattern logic and action model logic} \label{sect:tp}

\subsection{Language} \label{sec.language}

Given are a finite set of \emph{agents} $A$ and a set of \emph{propositional variables} $P \subseteq P' \times A$, where $P'$ is a countable set. For $B \subseteq A$ and $Q \subseteq P$, $Q \inter (P' \times B)$ is denoted $Q_B$ (where  $Q_a$ is $Q_{\{a\}}$), and $(p,a) \in P$ is denoted $p_a$. The set $P_a$ consists of the \emph{local variables} of agent $a$. In this work we consider the following languages.

\begin{definition}[Language] Given $A$ and $P$, the language $\lang^{\times\circ}$ is defined by BNF (where $p_a \in P$, $B \subseteq A$):
\[\begin{array}{ll} \phi := p_a \mid \neg \phi \mid \phi \et \phi \mid D_B \phi \mid [\cp,\cg]\phi \mid [\U,e]\phi
\end{array}\]
where $(\cp,\cg)$ and $(\U,e)$ are structures defined below, with $\cg\in\cp$ and $e$ in the domain of $\U$. Furthermore, $\lang^\circ$ is the language without $[\U,e]\phi$, $\lang^\times$ without $[\cp,\cg]\phi$, and $\lang^-$ without either.
\end{definition}
\emph{Epistemic} formula $D_B \phi$ is read as `the agents in $B$ have distributed knowledge of $\phi$'. We write $K_a \phi$ for $D_{\{a\}}\phi$, for `agent $a$ knows $\phi$'. \emph{Dynamic} formula $[\cp,\cg]\phi$ means `after execution of communication graph $\cg$ from communication pattern $\cp$, $\phi$ is true', and $[\U,e]\phi$ means `after execution of action $e$ from action model $\U$, $\phi$ is true'. Dynamic modalities will be interpreted as updates of epistemic models.

By notational abbreviation we define $[\U]\phi := \Et_{e \in E} [\U,e]\phi$ and $[\cp]\phi := \Et_{\cg\in\cp}[\cp,\cg]\phi$. 
The \emph{modal depth} of a formula $\phi\in\lang^{\circ\times}$ is inductively defined as: $md(p_a)=0$, $md(\neg\phi)= md(\phi)$, $md(\phi\et\psi) = \max \{md(\phi),md(\psi)\}$, $md(D_B\phi) := md(\phi)+1$, $md([\cp,\cg]\phi) := md(\phi)$, $md([\U,e]\phi) = md(\U) + md(\phi)$, where $md(\U) = \max \{ md(\pre(f)) \mid f \in E \}$. In $md(\U)$, the formulas $\pre(f)$ are defined below.

If $P$ is finite and $Q \subseteq P$, \emph{description} $\delta_Q$ (of valuation $Q$) is defined as $\Et_{p_a\in Q} p_a \et \Et_{p_a \in P{\setminus}Q} \neg p_a$. If $P$ is infinite and $Q \subseteq Q' \subset P$ are finite subsets of $P$, description $\delta_{Q,Q'}$ is defined as $\Et_{p_a\in Q} p_a \et \Et_{p_a \in Q'{\setminus}Q} \neg p_a$. 

\subsection{Structures} \label{sec.structures}


\begin{definition}[Epistemic model]
An {\em epistemic model} $M$ is a triple $(W,\sim,L)$, where for all $a \in A$, $\sim_a$ is an {\em equivalence relation} on the \emph{domain} $W$ (also denoted $\domain(M)$) consisting of \emph{states} (or \emph{worlds}), and where $L: W \imp \power(P)$ is the \emph{valuation} (function). For $\Inter_{a \in B} \sim_a$ we write $\sim_B$, and for $\{ w' \in W \mid w' \sim_a w\}$ we write $[w]_a$. We further require epistemic models to be \emph{local}: for all $a \in A$ and $v,w \in W$, $v \sim_a w$ implies $L(v)_a =L(w)_a$; if for all $a,v,w$ also $L(v)_a =L(w)_a$ implies $v \sim_a w$, it is a \emph{(static) interpreted system}. 
\end{definition}
An epistemic model encodes uncertainty among the agents about the value of other agents' local variables and about the knowledge of other agents.

\begin{definition}[Communication pattern]
A {\em communication graph} $\cg$ is a reflexive binary relation on the set of agents $A$, that is, $\cg \in \power(A \times A)$ and such that for all $a \in A$, $(a,a) \in \cg$. A {\em communication pattern} $\cp$ is a set of communication graphs, that is, $\cp \subseteq \power(A \times A)$.
\end{definition}
Expression $(a,b) \in \cg$ means that the message sent by $a$ is received by $b$.
For $(a,b) \in \cg$ we write $a \cg b$. We let $\cg b := \{a \in A \mid a \cg b\}$, $\cg B := \Union_{b \in B} \cg b$, and $\cg' B \equiv \cg B$ if $\cg' a = \cg a$ for all $a \in B$. The \emph{identity relation} $I$ is $\{(a,a) \mid a \in A\}$. The \emph{universal relation} $U$ is $A \times A$. 
A communication graph is a reflexive relation, because we assume that an agent always receives her own message. But not every other agent may receive the message. We could alternatively have defined a communication pattern as a structure with equivalence relations $\sim_a$ for each agent, namely by defining that $\cg \sim_a \cg'$ iff $\cg a = \cg' a$, as in \cite{diego:2021}.

\begin{definition}[Action model]
An {\em action model} $\U = (E, \sim, \pre)$ consists of a {\em domain} $E$ of {\em actions}, an {\em accessibility function} $\sim \ : A \imp {\mathcal P}(E \times E)$, where each $\sim_a$ is an equivalence relation, and a {\em precondition function} $\pre: E \imp \lang^-$.
\end{definition}
An action model \cite{baltagetal:1998} is a structure like an epistemic model but with a precondition function, associating a formula with each state.
The restriction to language $\lang^-$ for preconditions excuses us from explanations involving mutual recursion. 

For all the above structures $X$ we also consider \emph{pointed} and \emph{multi-pointed} versions that are pairs $(X,x)$ with $x \in X$ (or $x \in \domain(X)$) resp.\ $(X,Y)$ with $Y \subseteq X$ ($Y \subseteq \domain(X)$), so we have pointed epistemic models $(M,w)$, multi-pointed action models $(\U,T)$, etcetera.

Communication patterns are fairly novel in dynamic epistemic logic. We note that similar structures or modalities were proposed in \cite{AgotnesW17} (resolving distributed knowledge), in \cite{Baltag20} (reading events), and in \cite{diego:2021} (communication patterns). The communication patterns in \cite{diego:2021} have preconditions, just as action models. The reading events in \cite{Baltag20} and resolution in \cite{AgotnesW17} are communication patterns without uncertainty over the reception of messages. Then again, communication patterns permit less uncertainty than the arbitrary reading events in \cite{Baltag20}. These differences are discussed in \cite{cdrv:2022}. Examples are given in Sect.~\ref{sect:iamis}.

One can update an epistemic model with a communication pattern and one can also update an epistemic model with an action model. The updated epistemic model encodes how the knowledge has changed after agents have informed each other according to the update.

Given an epistemic model $M=(W,\sim,L)$ and a communication pattern $\cp$, the \emph{updated} epistemic model $M \odot \cp = (\dot W, \dot\sim, \dot L)$ (the \emph{update} of $M$ with $\cp$) is defined as: 
\[\begin{array}{lcl}
\dot W & = & W \times \cp \\
(w,\cg) \dot\sim_a (w',\cg') & \text{iff} &  w \sim_{\cg a} w' \text{ and } \cg a = \cg' a \\
\dot L(w,\cg) & = & L(w)
\end{array}\]
The relation $\dot\sim_a$ is the intersection $\sim_{\cg a}$ of the relations of all agents from which $a$ received messages.

Given an epistemic model $M = (W,\sim,L)$ and an action model $\U = (E, \sim, \pre)$, the updated epistemic model $M \otimes \U = (W^\times, \sim^\times, L^\times)$ is defined as: \[\begin{array}{lcl}  W^\times & \ \ = \ \ & \{ (v,f) \mid v \in W, f \in E, \text{ and } M,v \models \pre(f) \} \\  (v,f) \sim^\times_a (v',f') & \text{iff} & v \sim_a v' \text{ and } f \sim_a f'  \\  L^\times(v,f) & = & L(v)
\end{array} \]
The satisfaction relation $\models$ to determine $M,v \models \pre(f)$ is defined below, by mutual recursion.

In order to compare the information content of epistemic models we need the notions of \emph{(collective) bisimulation} and \emph{bounded (collective) bisimulation} (\emph{$n$-bisimulation}) \cite{blackburnetal:2001,Roelofsen07}. 

\begin{definition}[Collective bisimulation]
A relation $Z$ between the domains of epistemic models $M = (W,\sim,L)$ and $M'=(W',\sim',L')$ is a {\em (collective) bisimulation}, notation $Z: M \bisim M'$, if for all $(w,w') \in Z$:
\begin{itemize}
\item {\bf atoms}: for all $p_a \in P$, $p_a \in L(w)$ iff $p_a \in L'(w')$;
\item {\bf forth}: for all nonempty $B \subseteq A$ and for all $v \in W$, if $w \sim_B v$ then there is $v'\in W'$ such that $(v,v') \in Z$ and $w' \sim_B v'$;
\item {\bf back}: for all nonempty $B \subseteq A$ and for all $v' \in W'$, if $w' \sim_B v'$ then there is $v\in W$ such that $(v,v') \in Z$ and $w \sim_B v$.
\end{itemize}
We additionally define a  \emph{collective bisimulation bounded by $n$}, as a set of relations $Z^0 \supseteq Z^1 \dots \supseteq Z^n$ of $i$-bisimulations for $0 \leq i \leq n$. Relation $Z^0$ merely satisfies {\bf atoms},
and for all $(w,w') \in Z^{n+1}$:
\begin{itemize}
\item {\bf atoms}: for all $p_a \in P$, $p_a \in L(w)$ iff $p_a \in L'(w')$;
\item {\bf forth}-$(n+1)$:  for all nonempty $B \subseteq A$ and for all $v \in W$, if $w \sim_B v$ then there is $v'\in W'$ such that $(v,v') \in Z^n$ and $w' \sim_B v'$.
\item {\bf back}-$(n+1)$:  for all nonempty $B \subseteq A$ and for all $v' \in W'$, if $w' \sim_B v'$ then there is $v\in W$ such that $(v,v') \in Z^n$ and $w \sim_B v$.
\end{itemize}
If there is a bisimulation $Z$ between $M$ and $M'$ we write $M \bisim M'$, and if there is one containing $(w,w')$ we write $(M,w)\bisim (M',w')$. We then say that $M$ and $M'$, respectively $(M,w)$ and $(M',w')$, are \emph{bisimilar}. If $Z$ is bounded by $n$ we write $(M,w)\bisim^n (M',w')$ and we say that $(M,w)$ and $(M',w')$ are \emph{$n$-bisimilar}.
\end{definition}

Bounded bisimulations are used to compare models $(M,w)$ and $(M',w')$ up to a depth $n$ from the respective points $w$ and $w'$. Collective $n$-bisimilarity implies that both models satisfy the same $\lang^-$ formulas of modal depth at most $n$, as a minor variation of the standard result in \cite{blackburnetal:2001}. 

To compare dynamic modalities we define \emph{updates} and {\em update expressivity}.
\begin{definition}[Update, update expressivity]
An \emph{update} (or \emph{update relation}) is a binary relation $X$ on a class of pointed epistemic models. Given updates $X$ and $Y$, {\em $X$ is update equivalent to $Y$}, if for all pointed epistemic models $(M,w)$ the update of $(M,w)$ with $X$ is collectively bisimilar to the update of $(M,w)$ with $Y$. Update modalities $[X]$ and $[Y]$ are update equivalent, if $X$ and $Y$ are update equivalent. (For more refined notions see \cite{hvdetal.aus:2020}.)

A language $\lang$ is {\em at least as update expressive as} $\lang'$ if for every update modality $[X]$ of $\lang'$ there is an update modality $[Y]$ of $\lang$ such that $X$ is update equivalent to $Y$. Language $\lang$ is {\em equally update expressive as} $\lang'$ (or `as update expressive as'), if $\lang$ is at least as update expressive as $\lang'$ and $\lang'$ is at least as update expressive as $\lang$. Language $\lang$ is {\em (strictly) more update expressive than} $\lang'$, if $\lang$ is at least as update expressive as $\lang'$ and $\lang'$ is not at least as update expressive as $\lang$. Languages $\lang$ and $\lang'$ are incomparable in update expressivity if if $\lang$ is not at least as update expressive as $\lang'$ and $\lang'$ is not at least as update expressive as $\lang$.
\end{definition}

\subsection{Semantics}

\begin{definition}[Semantics on epistemic models] Given $M = (W,\sim,L)$ and $w \in W$, the \emph{satisfaction relation} $\models$ is defined by induction on $\phi\in\lang^{\times\circ}$ (where $p \in P$, $a \in A$, $B \subseteq A$, $(\cp,\cg)$ a pointed communication pattern and $(\U,e)$ a pointed action model).
\[ \begin{array}{lcl}
M,w \models p_a & \text{iff} & p_a \in L(w)\\
M,w \models \neg\phi & \text{iff} & M,w \not\models \phi\\
M,w \models \phi\et\psi & \text{iff} & M,w \models \phi \text{ and } M,w \models \psi \\
M,w \models D_B \phi & \text{iff} & M,v \models \phi \text{ for all } v \sim_B w \\
M,w \models [\cp,\cg]\phi & \text{iff} & M \odot \cp, (w,\cg) \models \phi \\
M,w \models [\U,e]\phi & \text{iff} & M,w \models\pre(e) \text{ implies } M \otimes \U, (w,e) \models \phi
\end{array} \]
Formula $\phi$ is {\em valid on $M$} iff for all $w \in W$, $M,w \models \phi$; formula $\phi$ is {\em valid} iff for all $(M,w)$, $M,w \models \phi$. 
\end{definition}

The (required) locality of epistemic models causes distributed knowledge to have slightly different properties in our semantics. In the standard semantics of distributed knowledge $D_B \phi \eq \phi$ is invalid for any $B \subseteq A$. Whereas in our semantics $D_A \phi \eq \phi$ is valid although $D_B \phi \eq \phi$ for $B \subset A$ remains invalid. 

A complete axiomatization of the validities of $\lang^\circ$ (communication pattern logic), reducing the dynamics, is given in \cite{cdrv:2022} (similar to \cite{AgotnesW17,Baltag20}). A complete axiomatization of the validities of $\lang^\times$ (action model logic), reducing the dynamics, is given in \cite{baltagetal:1998}. The language $\lang^{\times\circ}$ is not of independent interest.

\section{Induced action models for interpreted systems} \label{sect:iamis}

In this section, let $P$ be finite. From each communication pattern we will construct an \emph{induced action model}. We will show that communication patterns are update equivalent to induced action models when executed in an interpreted system. However, the update of an interpreted system with a communication pattern may not be an interpreted system, and the update of an epistemic model that is not an interpreted system with a communication pattern may not have the same update effect as its induced action model, of which we will give an example.

\begin{definition}[Action model induced by a communication pattern] \label{def.induced} Given a communication pattern $\cp$, define \emph{induced action model} $\U(\cp) = (E,\sim,\pre)$ as follows (where $\cg,\cg' \in \cp$, $Q,Q' \subseteq P$, $a \in A$).
\[ \begin{array}{lcl}
E & = & \cp \times \power(P) \\
(\cg,Q) \sim_a (\cg',Q') & \text{iff} & \cg a = \cg' a \text{ and } Q_{\cg a} = Q'_{\cg' a} \\
\pre(\cg,Q) & = & \delta_Q
\end{array}\]
\end{definition}
Informally, this says that two actions are indistinguishable for an agent if the agent receives messages from the same agents ($\cg a = \cg' a$) and if the messages it receives from those agents are the same ($Q_{\cg a} = Q'_{\cg' a}$). As $\cp$ and $P$ are finite, $\U(\cp)$ has a finite domain, so that modality $[\U(\cp)]$ is in $\lang^\times$. The size of the action model $\U(\cp)$ is $|\cp \times \power(P)| = |\cp| \cdot 2^{|P|}$. Therefore, $\U(\cp)$ is exponentially larger than $\cp$.

\begin{proposition} \label{prop.xxxx}
Let an interpreted system $M$ and $\cp$ be given. Then $M \odot \cp$ is bisimilar to $M \otimes \U(\cp)$.
\end{proposition}

\begin{proof}
Let $M = (W,\sim,L)$. Define the following relation $Z$ between (the domains of)  $M \odot \cp$ and $M \otimes \U(\cp)$: $Z : (w,\cg) \mapsto (w,(\cg, L(w)))$. We show that $Z$ defines a bisimulation. 

Let $((w,\cg), (w,\cg, L(w)) \in Z$.

{\bf atoms}:  Straightforwardly, $\dot L(w,\cg)=L(w)= L^\times(w,(\cg, L(w)))$. 

{\bf forth}:  Assume $(w,\cg) \dot\sim_B (v,S)$. We claim that $(v,(S,L(v)))$ is the required witness to show {\bf forth}. Obviously $((v,S),(v,(S,L(v))) \in Z$. We also have:

\medskip

\noindent
$(w,\cg) \dot\sim_B (v,S)$ \quad $\Eq $ \\
for all $a \in B$, $(w,\cg) \dot\sim_a (v,S)$ \quad $\Eq$ \hfill by definition of $\dot\sim_a$ \\
for all $a \in B$, $w \sim_{\cg a} v$ and $\cg a= S a$ \quad $\Eq $ \hfill (*) \\
for all $a \in B$, $w \sim_a v$, $\cg a = S a$, and $L(w)_{\cg a} = L(v)_{S a}$ \quad $\Eq $ \\
for all $a \in B$, $w \sim_a v$ and $(\cg,L(w)) \sim_a (S,L(v))$ \quad $\Eq $ \\
for all $a \in B$, $(w,(\cg,L(w))) \sim^\times_a (v,(S,L(v)))$ \quad $\Eq $ \\
$(w,(\cg,L(w))) \sim^\times_B (v,(S,L(v)))$.

\medskip

(*): As $M$ is an interpreted system, for all agents $b \in \cg a$, $w \sim_b v$ iff $L(w)_b = L(v)_b$, in other words: $w \sim_{\cg a} v$ iff $L(w)_{\cg a} = L(v)_{S a}$. As in particular $a \in \cg a$, $w \sim_a v$ on the right-hand side of the equation also follows from $L(w)_{\cg a} = L(v)_{S a}$.

{\bf back}: Similar to {\bf forth}.
\end{proof}


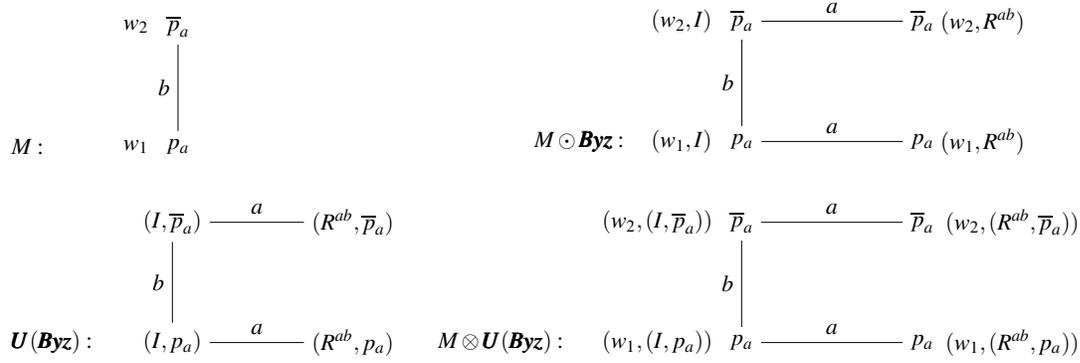
\begin{figure}
\scalebox{.8}{
\begin{tikzpicture}
\node (m) at (-6.5,0) {$M:$};
\node (bm00) at (-4.7,0) {$w_1$};
\node (am01) at (-4.7,2) {$w_2$};
\node (m00) at (-4,0) {$p_a$};
\node (m01) at (-4,2) {$\overline{p}_a$};
\draw[-] (m00) -- node[left] {$b$} (m01);
\end{tikzpicture}
\quad \hspace{4.9cm}
\begin{tikzpicture}
\node (m) at (-.7,0) {$M \odot \byz:$};
\node (b00) at (1,0) {$(w_1, I)$};
\node (b10) at (6,0) {$(w_1, R^{ab})$};
\node (a01) at (1,2) {$(w_2, I)$};
\node (a11) at (6,2) {$(w_2,R^{ab})$};
\node (00) at (2,0) {$p_a$};
\node (10) at (5,0) {$p_a$};
\node (01) at (2,2) {$\overline{p}_a$};
\node (11) at (5,2) {$\overline{p}_a$};
\draw[-] (00) -- node[above] {$a$} (10);
\draw[-] (00) -- node[left] {$b$} (01);
\draw[-] (01) -- node[above] {$a$} (11);
\end{tikzpicture}
}

\bigskip

\scalebox{.8}{
\begin{tikzpicture}
\node (m) at (-2,0) {$\U(\byz):$};
\node (00) at (0,0) {$(I, p_a)$};
\node (10) at (3,0) {$(R^{ab}, p_a)$};
\node (01) at (0,2) {$(I, \overline{p}_a)$};
\node (11) at (3,2) {$(R^{ab}, \overline{p}_a)$};
\draw[-] (00) -- node[above] {$a$} (10);
\draw[-] (00) -- node[left] {$b$} (01);
\draw[-] (01) -- node[above] {$a$} (11);
\end{tikzpicture}
\quad
\begin{tikzpicture}
\node (m) at (-2,0) {$M \otimes \U(\byz):$};
\node (b00) at (.6,0) {$(w_1, (I, p_a))$};
\node (b10) at (6.5,0) {$(w_1, (R^{ab}, p_a))$};
\node (a01) at (.6,2) {$(w_2, (I, \overline{p}_a))$};
\node (a11) at (6.5,2) {$(w_2, (R^{ab}, \overline{p}_a))$};
\node (00) at (2,0) {$p_a$};
\node (10) at (5,0) {$p_a$};
\node (01) at (2,2) {$\overline{p}_a$};
\node (11) at (5,2) {$\overline{p}_a$};
\draw[-] (00) -- node[above] {$a$} (10);
\draw[-] (00) -- node[left] {$b$} (01);
\draw[-] (01) -- node[above] {$a$} (11);
\end{tikzpicture}
}
\caption{Communication pattern and action model for Byzantine Generals}
\label{figure.byz}
\end{figure}

\begin{example}[Byzantine generals] \label{ex.uq}
Byzantine attack \cite{lamportetal:1982,DworkM90} is a communication pattern given in \cite{diego:2021}. Let $A = \{a,b\}$ and $P = \{p_a\}$. Generals $a$ and $b$ wish to schedule an attack, where $b$ desires to learn whether $a$ wants to `attack at dawn' ($p_a$) or `attack at noon' ($\neg p_a$). General $a$ now sends her decision to general $b$ in a message that may fail to arrive. This fits the communication pattern $\cp = \{I,R^{ab}\}$ where $R^{ab} = I \union \{(a,b)\}$, which models that $a$ is uncertain whether her message has been received by $b$. In this instantiation of Byzantine generals, general $b$ has no local variable. 

The communication pattern $\byz = \{I,R^{ab}\}$ where $R^{ab} = I \union \{(a,b)\}$. We have that $I a = R^{ab} a = \{a\}$ whereas $I b = \{b\}$ and $R^{ab} b = \{a,b\}$ (see also  \cite[Figure 1]{diego:2021} and \cite[Example 7]{cdrv:2022}). 

Fig.~\ref{figure.byz} depicts the initial epistemic model $M$ wherein agent $b$ is uncertain about the value of a variable $p_a$ of agent $a$, the updated model $M \odot \byz$, the action model $\U(\byz)$, and updated model $M \otimes \U(\byz)$. The states in epistemic models are also labelled with valuations, where $p_a$ stands for $\{p_a\}$ and $\overline{p}_a$ stands for $\emptyset$. Model $M$ is an interpreted system in the vacuous sense that if agent $b$ were to have local variables we could assume their value to be the same in both states. In $\U(\byz)$, the precondition of actions $(I, \{p_a\})$ and $(R^{ab}, \{p_a\})$ is $p_a$, and that of actions $(I, \emptyset)$ and $(R^{ab}, \emptyset)$ is $\neg p_a$. (In the figure, for visual consistency, these actions are written as $(I, p_a)$, $(R^{ab}, p_a)$, $(I, \overline{p}_a)$, and $(R^{ab}, \overline{p}_a)$.) Model $M \otimes \U(\byz)$ is bisimilar, as required, to $M\odot \byz$ and even isomorphic.
\end{example}

When model $M$ is an interpreted system, $M \odot \cp$ may not be an interpreted system, as, in a way, $M \odot \byz$ in Example~\ref{ex.uq}. If agent $b$ were to have local variables, their value would be the same in $w_1$ and in $w_2$ and thus also in the four worlds of the updated model. But now agent $b$ has three equivalence classes. It is therefore no longer an interpreted system.

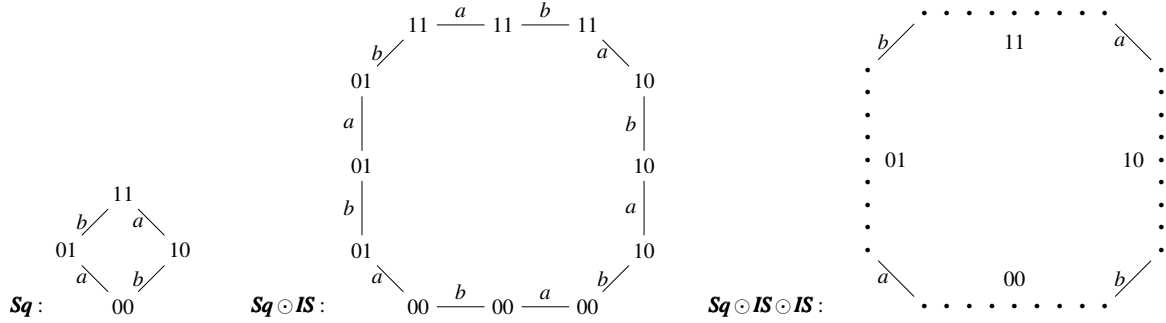
\begin{figure}
\scalebox{.75}{
\begin{tikzpicture}
\node (m) at (-.7,0) {$\Sq:$};
\node (00) at (1,0) {$00$};
\node (01) at (0,1) {$01$};
\node (10) at (2,1) {$10$};
\node (11) at (1,2) {$11$};
\draw[-] (00) -- node[left] {$a$} (01);
\draw[-] (10) -- node[left] {$a$} (11);
\draw[-] (00) -- node[left] {$b$} (10);
\draw[-] (01) -- node[left] {$b$} (11);
\end{tikzpicture}
}
\quad
\scalebox{.75}{
\begin{tikzpicture}
\node (m) at (-1.3,0) {$\Sq \odot \IS:$};
\node (00b) at (1,0) {$00$};
\node (00ab) at (2.5,0) {$00$};
\node (00a) at (4,0) {$00$};
\node (01b) at (0,1) {$01$};
\node (01ab) at (0,2.5) {$01$};
\node (01a) at (0,4) {$01$};
\node (10a) at (5,1) {$10$};
\node (10ab) at (5,2.5) {$10$};
\node (10b) at (5,4) {$10$};
\node (11a) at (1,5) {$11$};
\node (11ab) at (2.5,5) {$11$};
\node (11b) at (4,5) {$11$};
\draw[-] (00b) -- node[left] {$a$} (01b);
\draw[-] (00b) -- node[above] {$b$} (00ab);
\draw[-] (00ab) -- node[above] {$a$} (00a);
\draw[-] (10b) -- node[left] {$a$} (11b);
\draw[-] (10b) -- node[left] {$b$} (10ab);
\draw[-] (10ab) -- node[left] {$a$} (10a);
\draw[-] (00a) -- node[left] {$b$} (10a);
\draw[-] (01a) -- node[left] {$b$} (11a);
\draw[-] (01a) -- node[left] {$a$} (01ab);
\draw[-] (01ab) -- node[left] {$b$} (01b);
\draw[-] (11a) -- node[above] {$a$} (11ab);
\draw[-] (11ab) -- node[above] {$b$} (11b);
\end{tikzpicture}
}\quad
\scalebox{.75}{
\begin{tikzpicture}
\node (m) at (-1.8,0) {$\Sq \odot \IS \odot \IS:$};
\node (00) at (2.6,0.5) {$00$};
\node (11) at (2.6,4.7) {$11$};
\node (11) at (.5,2.6) {$01$};
\node (11) at (4.7,2.6) {$10$};
\node (00b) at (1,0) {\tiny $\bullet$};
\node (00b2) at (1.4,0) {\tiny $\bullet$};
\node (00b3) at (1.8,0) {\tiny $\bullet$};
\node (00b4) at (2.2,0) {\tiny $\bullet$};
\node (00b5) at (2.6,0) {\tiny $\bullet$};
\node (00b6) at (3.0,0) {\tiny $\bullet$};
\node (00b7) at (3.4,0) {\tiny $\bullet$};
\node (00b8) at (3.8,0) {\tiny $\bullet$};
\node (00b9) at (4.2,0) {\tiny $\bullet$};
\node (01a) at (0,1) {\tiny $\bullet$};
\node (01a2) at (0,1.4) {\tiny $\bullet$};
\node (01a3) at (0,1.8) {\tiny $\bullet$};
\node (01a4) at (0,2.2) {\tiny $\bullet$};
\node (01a5) at (0,2.6) {\tiny $\bullet$};
\node (01a6) at (0,3) {\tiny $\bullet$};
\node (01a7) at (0,3.4) {\tiny $\bullet$};
\node (01a8) at (0,3.8) {\tiny $\bullet$};
\node (01a9) at (0,4.2) {\tiny $\bullet$};
\node (10a) at (5.2,1) {\tiny $\bullet$};
\node (10a2) at (5.2,1.4) {\tiny $\bullet$};
\node (10a3) at (5.2,1.8) {\tiny $\bullet$};
\node (10a4) at (5.2,2.2) {\tiny $\bullet$};
\node (10a5) at (5.2,2.6) {\tiny $\bullet$};
\node (10a6) at (5.2,3) {\tiny $\bullet$};
\node (10a7) at (5.2,3.4) {\tiny $\bullet$};
\node (10a8) at (5.2,3.8) {\tiny $\bullet$};
\node (10a9) at (5.2,4.2) {\tiny $\bullet$};
\node (11b) at (1,5.2) {\tiny $\bullet$};
\node (11b2) at (1.4,5.2) {\tiny $\bullet$};
\node (11b3) at (1.8,5.2) {\tiny $\bullet$};
\node (11b4) at (2.2,5.2) {\tiny $\bullet$};
\node (11b5) at (2.6,5.2) {\tiny $\bullet$};
\node (11b6) at (3.0,5.2) {\tiny $\bullet$};
\node (11b7) at (3.4,5.2) {\tiny $\bullet$};
\node (11b8) at (3.8,5.2) {\tiny $\bullet$};
\node (11b9) at (4.2,5.2) {\tiny $\bullet$};
\draw[-] (00b) -- node[left] {$a$} (01a);
\draw[-] (10a9) -- node[left] {$a$} (11b9);
\draw[-] (00b9) -- node[left] {$b$} (10a);
\draw[-] (01a9) -- node[left] {$b$} (11b);
\end{tikzpicture}
}

\caption{Iterated immediate snapshot for two agents $a,b$. In world 10 local variable $p_a$ is true and $p_b$ is false (a slightly simpler depiction than $p_a\overline{p}_b$), etcetera. In $\Sq \odot \IS$ and $\Sq \odot \IS \odot \IS$ it is implicit which communication graph is executed, and in $\Sq \odot \IS \odot \IS$ valuations are only indicated schematically.}
\label{fig.is}
\end{figure}

\begin{example}[Iterated Immediate Snapshot] \label{example.iis}
Consider the model $\Sq$ where $a$ knows the truth about $p_a$ and $b$ knows the truth about $p_b$. This is the interpreted system for two agents each having a single variable. We recall the immediate snapshot ($\IS$) \cite{herlihyetal:2013} for two agents $\{a,b\}$, defined as $\{\cg^{ab},\cg^{ba},U\}$, where $\cg^{ab} = I \union \{(a,b)\}$ and $\cg^{ba} = I \union \{(b,a)\}$. These three communication graphs, as points of $\IS$, are commonly denoted as \emph{schedules} consisting of \emph{concurrency classes} $a.b$, $b.a$, and $ab$, respectively. Fig.~\ref{fig.is} shows the models $\Sq$, $\Sq \odot\IS$, and $\Sq \odot \IS \odot \IS$. Lemma~\ref{lemma.below} below shows that iteration of $\IS$ preserves circularity, as in the figure.


It follows from Prop.~\ref{prop.xxxx} that $\Sq \odot\IS$ is bisimilar to $\Sq \otimes \U(\IS)$. However, $(\Sq \otimes \U(\IS)) \otimes \U(\IS)$ is not  bisimilar to $(\Sq \odot \IS) \odot \IS$ and these models therefore satisfy different formulas in comparable worlds. In view of Prop.~\ref{prop.xxxx} it is sufficient to show that $(\Sq \odot \IS) \otimes \U(\IS)$ is not  bisimilar to $(\Sq \odot \IS) \odot \IS$.

Consider the fragment

\medskip

\noindent
\begin{tikzpicture}
\node (11a) at (1,5) {$(11,\cg^{ba})$};
\node (11ab) at (3.5,5) {$(11,U)$};
\node (11b) at (6,5) {$(11,\cg^{ab})$};
\draw[-] (11a) -- node[above] {$a$} (11ab);
\draw[-] (11ab) -- node[above] {$b$} (11b);
\end{tikzpicture}

\medskip

\noindent
of model $\Sq \odot \IS$. This is the top row in Fig.~\ref{fig.is}. In the model $\Sq \odot \IS \odot \IS$ this becomes

\medskip

\noindent
\scalebox{.6}{
\begin{tikzpicture}
\node (11a-a) at (1,5) {$(11,\cg^{ba},\cg^{ba})$};
\node (11a-ab) at (4,5) {$(11,\cg^{ba},U)$};
\node (11a-b) at (7,5) {$(11,\cg^{ba},\cg^{ab})$};
\node (11ab-b) at (10,5) {$(11,U,\cg^{ab})$};
\node (11ab-ab) at (13,5) {$(11,U,U)$};
\node (11ab-a) at (16,5) {$(11,U,\cg^{ba})$};
\node (11b-a) at (19,5) {$(11,\cg^{ab},\cg^{ba})$};
\node (11b-ab) at (22,5) {$(11,\cg^{ab},U)$};
\node (11b-b) at (25,5) {$(11,\cg^{ab},\cg^{ab})$};
\draw[-] (11a-a) -- node[above] {$a$} (11a-ab);
\draw[-] (11a-ab) -- node[above] {$b$} (11a-b);
\draw[-] (11a-b) -- node[above] {$a$} (11ab-b);
\draw[-] (11ab-b) -- node[above] {$b$} (11ab-ab);
\draw[-] (11ab-ab) -- node[above] {$a$} (11ab-a);
\draw[-] (11ab-a) -- node[above] {$b$} (11b-a);
\draw[-] (11b-a) -- node[above] {$a$} (11b-ab);
\draw[-] (11b-ab) -- node[above] {$b$} (11b-b);
\end{tikzpicture}
}

\medskip

Let us now, instead, calculate $\Sq \odot \IS \otimes \U(\IS)$. Instead of $(11,\cg^{ba},\cg^{ba})\text{---}a\text{---}(11,\cg^{ba},U)$, we obtain $(11,\cg^{ba},(\cg^{ba},11))\text{---}a\text{---}(11,\cg^{ba},(U,11))$. Apart from this edge and other expected edges as above, we now obtain additional edges as below (where we also assume transitivity).

\medskip

\noindent
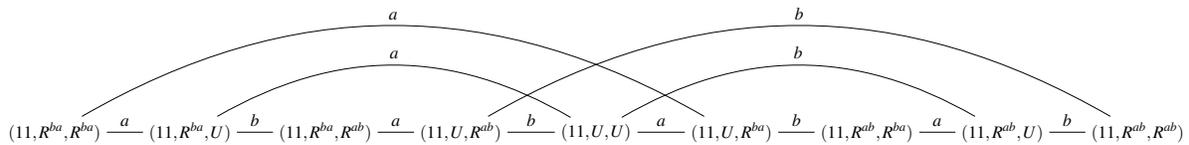
\scalebox{.6}{
\begin{tikzpicture}
\node (11a-a) at (1,5) {$(11,\cg^{ba},\cg^{ba})$};
\node (11a-ab) at (4,5) {$(11,\cg^{ba},U)$};
\node (11a-b) at (7,5) {$(11,\cg^{ba},\cg^{ab})$};
\node (11ab-b) at (10,5) {$(11,U,\cg^{ab})$};
\node (11ab-ab) at (13,5) {$(11,U,U)$};
\node (11ab-a) at (16,5) {$(11,U,\cg^{ba})$};
\node (11b-a) at (19,5) {$(11,\cg^{ab},\cg^{ba})$};
\node (11b-ab) at (22,5) {$(11,\cg^{ab},U)$};
\node (11b-b) at (25,5) {$(11,\cg^{ab},\cg^{ab})$};
\draw[-] (11a-a) -- node[above] {$a$} (11a-ab);
\draw[-] (11a-ab) -- node[above] {$b$} (11a-b);
\draw[-] (11a-b) -- node[above] {$a$} (11ab-b);
\draw[-] (11ab-b) -- node[above] {$b$} (11ab-ab);
\draw[-] (11ab-ab) -- node[above] {$a$} (11ab-a);
\draw[-] (11ab-a) -- node[above] {$b$} (11b-a);
\draw[-] (11b-a) -- node[above] {$a$} (11b-ab);
\draw[-] (11b-ab) -- node[above] {$b$} (11b-b);
\draw[-,bend left] (11a-ab) to node[above] {$a$} (11ab-ab);
\draw[-,bend left] (11a-a) to node[above] {$a$} (11ab-a);
\draw[-,bend left] (11ab-b) to node[above] {$b$} (11b-b);
\draw[-,bend left] (11ab-ab) to node[above] {$b$} (11b-ab);
\end{tikzpicture}
}

\medskip

For example, $(11, \cg^{ba}, (U,11)) \sim_a (11, U, (U,11))$, because by the semantics of action model execution, $(11,\cg^{ba}) \sim_a (11,U)$ in $\Sq \odot \IS$ and $(U,11) \sim_a (U,11)$ in $\U(\IS)$. Similarly, $(11, \cg^{ba}, (\cg^{ba},11)) \sim_a (11, U, (U,11))$, because $(11,\cg^{ba}) \sim_a (11,U)$ in $\Sq \odot \IS$ and $(\cg^{ba},11) \sim_a (U,11)$ in $\U(\IS)$, where the latter holds because $\cg^{ba} a = U a$ (namely $\{a,b\}$) and $11_{\cg^{ba} a} = 11_{U a}$ (namely $11_{\{ab\}}$, which is $11$).

Intuitively, in $\Sq \odot \IS \odot \IS$ the agents learn in the second round whether the communication succeeded in the previous, first, round. But in $\Sq \odot \IS \otimes \U(\IS)$ they do not learn this in the second round.

It is easy to see that $\Sq \odot \IS \odot \IS$ is not bisimilar to $\Sq \odot \IS \otimes \U(\IS)$ wherein we can reach states in the model with a different valuation in fewer steps. There are then distinguishing formulas, e.g., $\Sq \odot \IS \odot \IS, (11,U,\cg^{ba}) \not\models\M_a\M_b \neg p_a$, whereas $\Sq \odot \IS \otimes \U(\IS), (11,U,(\cg^{ba},11)) \models\M_a\M_b \neg p_a$.
\end{example}

\paragraph*{On squares and circles}  

A \emph{circular $ab$-chain} is an epistemic model consisting of an even number of worlds $0,\dots,2n-1$, where $n \in \Naturals$ with $n \geq 2$, and such that for all $i \leq n$, $2i \sim_a 2i+1$ and $2i \sim_b 2i-1$ (modulo $2n$).

\begin{lemma} \label{lemma.below}
Define $\Sq \odot \IS^0 := \Sq$ and $\Sq \odot \IS^{n+1} := (\Sq\odot \IS^n)\odot \IS$. For all $n \in \Naturals$, $\Sq \odot \IS^n$ is a circular $ab$-chain.
\end{lemma}
\begin{proof} We prove this by induction.

Model $\Sq$ is a (minimal) circular $ab$-chain.

Assuming that $\Sq \odot \IS^n$ is a circular $ab$-chain, take any world $w$ in that chain and let neighbouring worlds $w',w''$ be such that $w'  \sim_a w$ and $w \sim_b w''$ (where $w,w',w''$ have arbitrary valuation). We now execute $\IS$ once more. Consider the new worlds $(w,\cg^{ab}), (w,U), (w,\cg^{ba})$. Then:
\begin{itemize}
\item $(w',\cg^{ab}) \sim_a (w,\cg^{ab})$ because $\cg^{ab}a = \cg^{ab}a \ (=\{a\})$ and $w' \sim_a w$. No other world than $(w',\cg^{ab})$ is indistinguishable for $a$ from $(w,\cg^{ab})$. If $\cg \neq \cg^{ab}$ then $\cg a \neq \cg^{ab} a$ so $(w',\cg) \not\sim_a (w,\cg^{ab})$. If $v \neq w,w'$ then $v \not\sim_a w$ so $(v,\cg^{ab}) \not\sim_a (w,\cg^{ab})$.
\item $(w,\cg^{ba}) \sim_b (w'',\cg^{ba})$ because $\cg^{ba}b = \cg^{ba}b \ (=\{b\})$ and $w \sim_b w''$. Similarly to the previous case this is the unique indistinguishable other world in the updated model.
\item $(w,\cg^{ab}) \sim_b (w,U)$ because $\cg^{ab} b = U b \ (= \{a,b\})$ and $w \sim_{ab} w$. No other world than $(w,U)$ is indistinguishable for $b$ from $(w,\cg^{ab})$. We note that $\cg^{ba} b \neq \cg^{ab} b$ and $\cg^{ba} b \neq U b$, so  $(w,\cg^{ab}) \not\sim_b (w,\cg^{ba})$ and $(w,U) \not\sim_b (w,\cg^{ba})$. If $v \neq w$ then $v \not\sim_{ab} w$ so $(w,\cg^{ab}) \not\sim_b (v,\cg^{ab})$ and $(w,U) \not\sim_b (v,U)$.
\item $(w,\cg^{ba}) \sim_a (w,U)$ because $\cg^{ba} a = U a \ (= \{a,b\})$ and $w \sim_{ab} w$. Similarly to the previous case this is the unique indistinguishable other world in the updated model.
\end{itemize}
\end{proof}

This result is not surprising. In the corresponding representation as simplicial complexes, an application of $\IS$ is a so-called \emph{subdivision} \cite{herlihyetal:2013}. A circular $ab$-chain corresponds to a circular graph (1-dimensional complex) with alternating $a$ and $b$ nodes, such that each edge $a\text{---}b$ gets replaced by three edges $a\text{---}b\text{---}a\text{---}b$ at each iteration of $\IS$ (and duplicated nodes keeps their old labels).

\section{Communication patterns and action models are incomparable} \label{sect:inc}

\begin{proposition} \label{prop.prop} Communication pattern logic is not at least as update expressive as action model logic.
\end{proposition} 

\begin{proof}
We can prove this in different ways, which seems instructive.

First, in a public announcement, the environment may reveal something that cannot be revealed by the agents individually or jointly, such as the announcement whether $p_a \vel p_b$ in a model where $a$ knows whether $p_a$ and $b$ knows whether $p_b$.

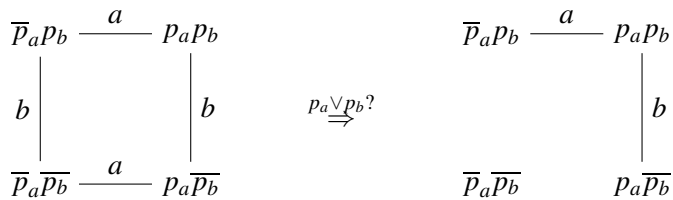
\medskip

\begin{tikzpicture}
\node (00) at (0,0) {$\overline{p}_a\overline{p_b}$};
\node (10) at (2,0) {$p_a\overline{p_b}$};
\node (01) at (0,2) {$\overline{p}_ap_b$};
\node (11) at (2,2) {$p_ap_b$};
\draw[-] (00) -- node[above] {$a$} (10);
\draw[-] (01) -- node[above] {$a$} (11);
\draw[-] (00) -- node[left] {$b$} (01);
\draw[-] (10) -- node[right] {$b$} (11);
\node (imp) at (4,1) {$\stackrel{p_a \vel p_b?}{\Imp}$};
\node (00r) at (6,0) {$\overline{p}_a\overline{p_b}$};
\node (10r) at (8,0) {$p_a\overline{p_b}$};
\node (01r) at (6,2) {$\overline{p}_ap_b$};
\node (11r) at (8,2) {$p_ap_b$};
\draw[-] (01r) -- node[above] {$a$} (11r);
\draw[-] (10r) -- node[right] {$b$} (11r);
\end{tikzpicture}

\medskip

Second, agents may choose to reveal some but not all of their local variables, such as, if $a$ knows whether $p_a$ and whether $q_a$, $a$ informing $b$ of the truth about $p_a$ but not about $q_a$.

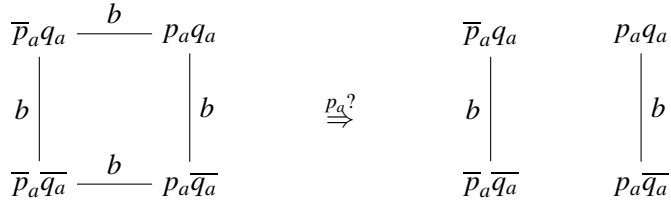
\begin{tikzpicture}
\node (00) at (0,0) {$\overline{p}_a\overline{q_a}$};
\node (10) at (2,0) {$p_a\overline{q_a}$};
\node (01) at (0,2) {$\overline{p}_aq_a$};
\node (11) at (2,2) {$p_aq_a$};
\draw[-] (00) -- node[above] {$b$} (10);
\draw[-] (01) -- node[above] {$b$} (11);
\draw[-] (00) -- node[left] {$b$} (01);
\draw[-] (10) -- node[right] {$b$} (11);
\node (imp) at (4,1) {$\stackrel{p_a?}{\Imp}$};
\node (00r) at (6,0) {$\overline{p}_a\overline{q_a}$};
\node (10r) at (8,0) {$p_a\overline{q_a}$};
\node (01r) at (6,2) {$\overline{p}_aq_a$};
\node (11r) at (8,2) {$p_aq_a$};
\draw[-] (00r) -- node[left] {$b$} (01r);
\draw[-] (10r) -- node[right] {$b$} (11r);
\end{tikzpicture}

\medskip

Third, there are action models that produce more uncertainty than any communication pattern. Here we should note that although the composition of two action models is again an action model (therefore, for all $\U,\U'$ there is a $\U''$, namely the composition of $\U$ and $\U'$, such that $[\U][\U']\phi \eq [\U'']\phi$), sequentially executing two communication patterns is typically not the same as executing a single communication pattern (it is not the case that for all $\cp,\cp'$ there is a $\cp''$ such that $[\cp][\cp']\phi \eq [\cp'']\phi$). For example, consider the models $\Sq$ and $\Sq \odot \IS \odot \IS$ (Example~\ref{example.iis}). The domain of model $\Sq$ consists of four worlds and that of $\Sq \odot \IS \odot \IS$ consists of $36$ worlds; it is nine times larger (and it is bisimulation minimal). Now there are only four different communication patterns for two agents (namely $I$, $\cg^{ba}$, $\cg^{ab}$, and $U$). So the maximum size of a model resulting from updating $\Sq$ with a communication pattern is $16$. Therefore there is no such communication pattern. In other words, there is no $\cp$ such that $\Sq \odot \IS \odot \IS$ is bisimilar to $\Sq \odot \cp$ which implies that there is no $\cp$ that has the same update effect as updating twice with $\IS$.

However, there is an action model $\U$ such that $\Sq \odot \IS \odot \IS$ is bisimilar to $\Sq \otimes \U$: its domain is the domain of $\Sq \odot \IS \odot \IS$; its relations for $a$ and $b$ are the relations for $a$ and $b$ on the model $\Sq \odot \IS \odot \IS$, and its preconditions are such that the precondition of a world $(ij,\cg,\cg')$ in the domain of $\Sq \odot \IS \odot \IS$ is the description $\delta_{ij}$ of the valuation $ij$. It is straightforward to see that $\Sq \odot \IS \odot \IS$ is even isomorphic to $\Sq \otimes \U$.

We conclude that there is no communication pattern that is update equivalent to this action model $\U$. Therefore, communication pattern logic is not at least as update expressive as action model logic.
\end{proof}

We continue by showing that action model logic is not at least as update expressive as communication pattern logic. If multi-pointed action models had not been allowed, a trivial way to show that, would have been to observe that single-pointed action models unlike communication patterns may not always be executable. Although true, that is not of interest. We prove this in a more meaningful way in the following Prop.~\ref{prop.ppppp}. Its proof assumes towards a contradiction that an action model $\U$ exists that is update equivalent to the communication pattern $\IS$, where we identify $\U$ with the multi-pointed action model $(\U,\domain(\U))$. We then compare the updates $\IS$ and $\U$ in epistemic model $\Sq \odot \IS^n$ for $n$ exceeding a function of the modal depth of any precondition of $\U$, and derive a contradiction. It may assist the reader to know that Ex.~\ref{example.iis} above replays this proof for $\U = \U(\IS)$ of which the action preconditions are booleans, such that $md(\U)=0$ and we can choose $n=1$.

\begin{proposition} \label{prop.ppppp}
Action model logic is not at least as update expressive as communication pattern logic.
\end{proposition} 

\begin{proof}
Suppose towards a contradiction that communication pattern $\IS$ is update equivalent to an action model $\U = (E,\sim,\pre)$. 

What do we know about $\U$? As $\IS$ is always executable, we may assume that the disjunction $\psi$ of all preconditions of actions $e$ in the domain $E$ of $\U$ is the triviality. Otherwise, given some model with $M,w \models\neg\psi$, we could update with $\IS$ but not with $\U$. Similarly, for any action $e$ in the domain $E$ of $\U$, there must be $f \in E$ such that $e \sim_a f$ and $\pre(e)=\pre(f)$ (and for agent $b$ there must be a $g \in E$ such that $g \sim_b f$ and $\pre(g)=\pre(f)$). Otherwise, consider a model $(M,w)$ that can only be updated with $(\U,e)$ (for which $M,w \models \pre(e)$). It can be updated with $(\IS,U)$ and also with $(\IS,\cg^{ba})$ resulting in states $(w,U)$ and $(w,\cg^{ba})$ satisfying different properties, as $(w,U)\sim_a(w,\cg^{ba})$ (because $U a = \cg^{ba} = \{a,b\}$), so that one or the other but not both can be bisimilar to $(w,e)$. Therefore, $\U$ must be a refinement of $\IS$ seen as a structure $\cg^{ab}\text{---}b\text{---}U\text{---}a\text{---}\cg^{ba}$. Its actions can therefore be assumed to have shape $(\cg,\phi)$ where $\cg$ is one of $\cg^{ab},U,\cg^{ba}$ and where $\phi\in\lang^\times$ is the precondition of that action, that is, $\pre(\cg,\phi) = \phi$.\footnote{By \emph{refinement} we mean that $\cg^{ab}$ can be seen as an equivalence class $\{ (\cg^{ab}, \phi) \mid (\cg^{ab}, \phi) \in \domain(\U) \}$, and similarly for $U$ and $\cg^{ba}$, where two such equivalence classes are indistinguishable for $a$ if there are $(\cg,\phi), (\cg',\phi')$ such that $(\cg,\phi)\sim_a (\cg',\phi')$, and similarly for $b$.}

The modality $[\U]$ is an operator in the language $\lang^\times$ and $|E|$ is finite, so that $md(\U) = \max \{md(\pre(e)) \mid e \in E\}$ is defined. Choose $n \in \Naturals$ with $n > \log_3 2(md(\U)+1)$ and consider $\Sq \odot \IS^n$, schematically depicted as:

\medskip

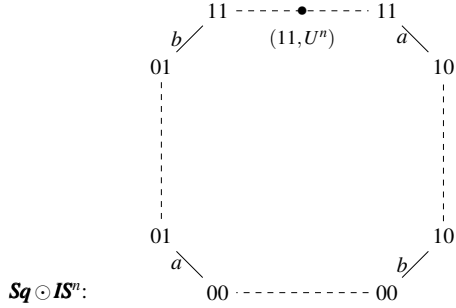
\scalebox{.75}{
\begin{tikzpicture}
\node (m) at (-2,0) {$\Sq \odot \IS^n$:};
\node (00b) at (1,0) {$00$};
\node (00a) at (4,0) {$00$};
\node (01b) at (0,1) {$01$};
\node (01a) at (0,4) {$01$};
\node (10a) at (5,1) {$10$};
\node (10b) at (5,4) {$10$};
\node (11a) at (1,5) {$11$};
\node (11b) at (4,5) {$11$};
\node (11m) at (2.5,5) {$\bullet$};
\node (11mb) at (2.5,4.5) {\small $(11,U^n)$};
\draw[-] (00b) -- node[left] {$a$} (01b);
\draw[-] (10b) -- node[left] {$a$} (11b);
\draw[-] (00a) -- node[left] {$b$} (10a);
\draw[-] (01a) -- node[left] {$b$} (11a);
\draw[-,dashed] (00a) -- (00b);
\draw[-,dashed] (10a) -- (10b);
\draw[-,dashed] (01a) -- (01b);
\draw[-,dashed] (11a) -- (11b);
\end{tikzpicture}
}

\medskip

\noindent
and concretely its three-action fragment:

\medskip

\noindent
\begin{tikzpicture}
\node (star) at (-1,5) {$(*):$};
\node (11a) at (1,5) {$(11,U^{n-1}\cg^{ba})$};
\node (11ab) at (4,5) {$(11,U^n)$};
\node (11b) at (7,5) {$(11,U^{n-1}\cg^{ab})$};
\draw[-] (11a) -- node[above] {$a$} (11ab);
\draw[-] (11ab) -- node[above] {$b$} (11b);
\end{tikzpicture}

\medskip

\noindent 
where world $(11,U^n)$ of $(*)$ is the same as the depicted world $(11,U^n)$ of $\Sq \odot \IS^n$. 

We can now justify the bound $n > \log_3 2(md(\U)+1)$. We need in the proof that the three worlds of $(*)$ satisfy the same actions of $\U$, and we guarantee that because they are bounded collectively bisimilar for an appropriate bound. Given $(11,U^n)$, the bound should exceed the modal depth of any possible precondition of any action in $\U$. That explains $md(\U)$. Plus one, as we need this to hold for the surrounding worlds too. That explains $md(\U)+1$. Twice that, $2 \cdot (md(\U)+1)$, is the required length of one side of the squarish model $\Sq \odot \IS^n$ with therefore $8 \cdot (md(\U)+1)$ worlds. Starting with four worlds, every iteration of $\IS$ multiplies the number of worlds by $3$. So we therefore want to iterate $\IS$ by some $n$ such that $4 \cdot 3^n >  8 \cdot (md(\U)+1)$, that is, $n > \log_3 2(md(\U)+1)$.

Consider $\Sq \odot \IS^n \otimes \U$.
Recalling what is known about $\U$, there must be an $e \in E$ such that $\Sq \odot \IS^n, (11,U^n) \models \pre(e)$. Also, there must be $f,g \in E$ with $e \sim_a f$ and $f \sim_b g$ and $\pre(e)=\pre(f)= \pre(e)$. Let $\pre(e)$ be $\theta$. These actions $e,f,g$ therefore have shape $(\cg^{ab},\theta)$, $(U,\theta)$, $(\cg^{ba},\theta)$  respectively. 

As $n > \log_3 2(md(\U)+1)$, the three worlds in $(*)$ are bounded collectively bisimilar: \[ (\Sq \odot \IS^n, (11,U^{n-1},\cg^{ba})) \bisim^{md(\U)+1} (\Sq \odot \IS^n, (11,U^n)) \bisim^{md(\U)+1} (\Sq \odot \IS^n, (11,U^{n-1},\cg^{ab})) \] As $md(\theta) \leq md(\U)$, all three worlds in $(*)$ satisfy $\theta$, so actions $e,f,g$ can be executed in all these worlds.

The model $\Sq \odot \IS^n \otimes \U$ therefore contains the submodel

\medskip

\noindent
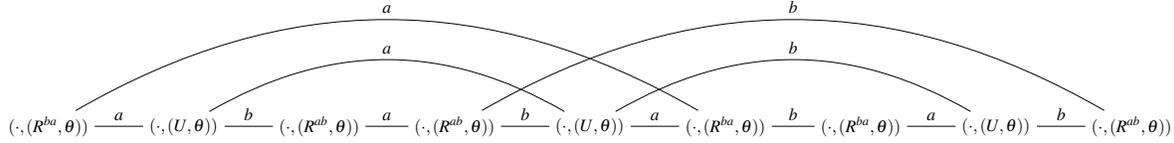
\scalebox{.6}{
\begin{tikzpicture}
\node (11a-a) at (1,5) {$(\cdot,(\cg^{ba},\theta))$};
\node (11a-ab) at (4,5) {$(\cdot,(U,\theta))$};
\node (11a-b) at (7,5) {$(\cdot,(\cg^{ab},\theta))$};
\node (11ab-b) at (10,5) {$(\cdot,(\cg^{ab},\theta))$};
\node (11ab-ab) at (13,5) {$(\cdot,(U,\theta))$};
\node (11ab-a) at (16,5) {$(\cdot,(\cg^{ba},\theta))$};
\node (11b-a) at (19,5) {$(\cdot,(\cg^{ba},\theta))$};
\node (11b-ab) at (22,5) {$(\cdot,(U,\theta))$};
\node (11b-b) at (25,5) {$(\cdot,(\cg^{ab},\theta))$};
\draw[-] (11a-a) -- node[above] {$a$} (11a-ab);
\draw[-] (11a-ab) -- node[above] {$b$} (11a-b);
\draw[-] (11a-b) -- node[above] {$a$} (11ab-b);
\draw[-] (11ab-b) -- node[above] {$b$} (11ab-ab);
\draw[-] (11ab-ab) -- node[above] {$a$} (11ab-a);
\draw[-] (11ab-a) -- node[above] {$b$} (11b-a);
\draw[-] (11b-a) -- node[above] {$a$} (11b-ab);
\draw[-] (11b-ab) -- node[above] {$b$} (11b-b);
\draw[-,bend left] (11a-ab) to node[above] {$a$} (11ab-ab);
\draw[-,bend left] (11a-a) to node[above] {$a$} (11ab-a);
\draw[-,bend left] (11ab-b) to node[above] {$b$} (11b-b);
\draw[-,bend left] (11ab-ab) to node[above] {$b$} (11b-ab);
\end{tikzpicture}
}

\medskip

\noindent 
wherein only some additional pairs for $\sim_a$ and $\sim_b$ are shown, and where from those shown we merely justify one as an example: for the leftmost and the middle worlds, we have that $(11, U^{n-1},\cg^{ba}, (\cg^{ba},\theta)) \sim_a (11, U^n, (U,\theta))$, because by the semantics of action model execution, $(11,U^{n-1},\cg^{ba}) \sim_a (11,U^n)$ in $\Sq \odot \IS^n$ and $(\cg^{ba},\theta) \sim_a (U,\theta)$ in $\U(\IS)$. Furthermore (unlike in Example~\ref{example.iis}), worlds $(\dots,(\cg,\theta))$ shown, may be indistinguishable for $a$ or $b$ from worlds $(\dots,(\cg,\xi))$ not shown, for actions $(\cg,\xi)$ with $\xi$ non-equivalent to $\theta$.

Consequently, $\Sq \odot \IS^n \otimes \U$ is not a circular $ab$-chain like $\Sq \odot \IS^{n+1}$ that locally looks like:

\medskip

\noindent
\scalebox{.6}{
\begin{tikzpicture}
\node (11a-a) at (1,5) {$(\cdot,\cg^{ba})$};
\node (11a-ab) at (4,5) {$(\cdot,U)$};
\node (11a-b) at (7,5) {$(\cdot,\cg^{ab})$};
\node (11ab-b) at (10,5) {$(\cdot,\cg^{ab})$};
\node (11ab-ab) at (13,5) {$(\cdot,U)$};
\node (11ab-a) at (16,5) {$(\cdot,\cg^{ba})$};
\node (11b-a) at (19,5) {$(\cdot,\cg^{ba})$};
\node (11b-ab) at (22,5) {$(\cdot,U)$};
\node (11b-b) at (25,5) {$(\cdot,\cg^{ab})$};
\draw[-] (11a-a) -- node[above] {$a$} (11a-ab);
\draw[-] (11a-ab) -- node[above] {$b$} (11a-b);
\draw[-] (11a-b) -- node[above] {$a$} (11ab-b);
\draw[-] (11ab-b) -- node[above] {$b$} (11ab-ab);
\draw[-] (11ab-ab) -- node[above] {$a$} (11ab-a);
\draw[-] (11ab-a) -- node[above] {$b$} (11b-a);
\draw[-] (11b-a) -- node[above] {$a$} (11b-ab);
\draw[-] (11b-ab) -- node[above] {$b$} (11b-b);
\end{tikzpicture}
}

\medskip

Now the assumption of update equivalence implies that $\Sq \odot \IS^{n+1}$ is collectively bisimilar to $\Sq \odot \IS^n \otimes \U$. The supposed bisimulation relation $Z$ linking $\Sq \odot \IS^{n+1}$ and $\Sq \odot \IS^n \otimes \U$ should therefore such that $Z: (w,\sigma,\cg) \mapsto (w,\sigma, (\cg,\pre(e))$ for all $w \in W$, $\sigma\in\IS^n$, and $e \in E$ with $\Sq \odot \IS^n,(w,\sigma) \models \pre(e)$, in particular the three worlds in $(*)$ and the $e,f,g$ above with preconditions $\theta$. On the other hand, clearly, a pair of worlds in this relation cannot be bisimilar, as the additional $a$-links and $b$-links allow shorter paths to a $01$-world. Differently said, as bounded bisimilarity implies the same truth value for formulas of at most that modal depth, the worlds in such a pair satisfy different formulas. (See Ex.~\ref{example.iis} for $n=1$.)

This contradicts our assumption that $\U$ is update equivalent to $\IS$ and thus concludes the proof.
\end{proof}

Prop.~\ref{prop.ppppp} holds for any countable set of local variables $P$. In the proof of Prop.~\ref{prop.ppppp} we only need two: $P = \{p_a,p_b\}$. When $P$ is countably infinite there is a shorter proof of Prop.~\ref{prop.ppppp}, given below.

\begin{proof} \label{ex.uq2}

Let $P$ be countably infinite. Suppose towards a contradiction that there is an action model $\U$ with $[\U]$ (or $[\U,e]$) in the logical language (so that the domain of $\U$ is necessarily finite) that is update equivalent to $\byz$. As $\U$ is finite and $P$ is countably infinite, there exists a $q_a \in P$ not occurring in any of the preconditions of the actions in the domain of $\U$. Now consider epistemic model $M''$ as in Example~\ref{ex.uq} but with $q_a$ true in $w_1$ and false in $w_2$ and with $p_a$ true in both worlds. When executing $\U$ in $M''$, the update $M'' \otimes \U$ will never get the required asymmetry of $M'' \odot \byz$, because any action (point) $e$ that is executable in $w_1$ is also executable in $w_2$, as for any $p \in P\setminus\{q_a\}$, $p \in L(w_1)$ iff $p \in L(w_2)$. In particular we therefore will have that $(w_1,I,\dots) \sim_b (w_2,I,\dots)$ iff $(w_1,\cg^{ab},\dots) \sim_b (w_2,\cg^{ab},\dots)$. (An argument involving $\sim_a$ and $\sim_b$ similar to the one in the proof of Prop.~\ref{prop.ppppp} is omitted for brevity.)

More simply said, if we were to execute $\U(\byz)$ of Example~\ref{ex.uq} in that model $M''$, the following model would result (as $p_a$ is true in $w_1$ and $w_2$, the alternatives with precondition $\neg p_a$ never execute):

\medskip

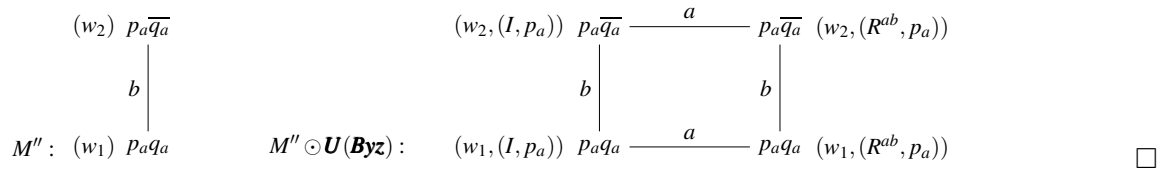
\scalebox{.8}{
\begin{tikzpicture}
\node (m) at (-5.9,0) {$M'':$};
\node (bm00) at (-4.9,0) {$(w_1)$};
\node (am01) at (-4.9,2) {$(w_2)$};
\node (m00) at (-4,0) {$p_aq_a$};
\node (m01) at (-4,2) {$p_a\overline{q_a}$};
\draw[-] (m00) -- node[left] {$b$} (m01);
\node (m) at (-.85,0) {$M'' \odot \U(\byz):$};
\node (b00) at (2,0) {$(w_1, (I, p_a))$};
\node (b10) at (8.2,0) {$(w_1, (R^{ab}, p_a))$};
\node (a01) at (2,2) {$(w_2, (I, p_a))$};
\node (a11) at (8.2,2) {$(w_2, (R^{ab}, p_a))$};
\node (00) at (3.5,0) {$p_aq_a$};
\node (10) at (6.5,0) {$p_aq_a$};
\node (01) at (3.5,2) {$p_a\overline{q_a}$};
\node (11) at (6.5,2) {$p_a\overline{q_a}$};
\draw[-] (00) -- node[above] {$a$} (10);
\draw[-] (00) -- node[left] {$b$} (01);
\draw[-] (01) -- node[above] {$a$} (11);
\draw[-] (10) -- node[left] {$b$} (11);
\end{tikzpicture}
}
\end{proof}

\begin{corollary}
Communication pattern logic and action model logic are incomparable in update expressivity.
\end{corollary}

\section{Communication patterns for history-based structures} \label{sect:hist}

Example~\ref{example.iis} demonstrated that interpreted systems are not closed under update with communication patterns. We therefore could not obtain a result for update expressivity for the class of interpreted systems. In this section we show that this is after all possible if we adjust the structures in which we execute updates and simultaneously adjust the definition of the update. In order to store the sequence of past events we generalize our epistemic models to history-based epistemic models \cite{jfaketal.lcc:2006,hvdetal.FAMAS:2013}. Simultaneously, we change the semantics of the update with a communication pattern namely by having this depend on the number of previous updates that already took place, what is known as the number of previous \emph{rounds} in an oblivious protocol arbitrarily often executing that communication pattern. The change consists in recording the information of previous rounds in designated history variables, that store the \emph{view} for each agent on all previous rounds. These variables are also local.

\begin{example} \label{example.history}
When updating epistemic model $(\Sq,11)$ with communication pattern $(\IS,\cg^{ab})$, we record that $\cg^{ab} a = \{a\}$ and $\cg^{ab} b = \{a,b\}$ in the resulting world $(11,\cg^{ab})$ by indexing these sets with the names of the agents, so as $\{a\}_a$ and $\{a,b\}_b$, that we write as $a_a$ and $ab_b$. These are local variables. Then, when updating $(\Sq \odot \IS, (11,\cg^{ab}))$ with $(\IS,\cg^{ba})$, we record the entire history so far for $a$ and $b$, where $a$ but not $b$ also receives $b$'s history of the previous round, as $((a,ab).ab)_a$ for agent $a$ and $(ab.b)_b$ for agent $b$. 

We explain the first. As $a$ receives information from $b$, and by default from itself, $\cg^{ba} a = \{a,b\}$, written as $ab$, is preceded by the list $(\{a\},\{a,b\})$ containing $\cg^{ab} a = \{a\}$ and $\cg^{ab} b = \{a,b\}$ of the previous round, which is written as $(a,ab)$. The expression $(a,ab).ab$ is the \emph{view} of agent $a$ on the history, which is a tree. This view is indexed with the name $a$ of the agent, such that $((a,ab).ab)_a$ is a local variable for agent $a$, wherein the views of $a$ and of $b$ in the previous round are lexicographically ordered.

And so on for every next round. Such history variables are designated local variables,  initially false.

We adapt the semantics of update $\odot$ by making history variables for a given round of communication true after the update representing that round. We name this semantics $\dot\odot$.
\begin{itemize}
\item in $(\Sq,11)$, variables $p_a$ and $p_b$ are true and all others false; 
\item in $(\Sq \dot\odot \IS, (11,\cg^{ab}))$, variables $p_a,p_b,a_a,ab_b$ are true and all others false; 
\item in $(\Sq \dot\odot \IS \dot\odot \IS, (11,\cg^{ab},\cg^{ba}))$, variables $p_a,p_b,a_a,ab_b,((a,ab).ab)_a, (ab.b)_b$ are true and \dots 
\end{itemize}
The $\dot\odot$ semantics is then closed for the class of interpreted systems. We proceed with formalities.
\end{example}

\begin{definition}[View, history variable]
Let a communication pattern $\cp$ be given. A \emph{history} is a member $\sigma \in \cp^*$ (a finite sequence of communication graphs in $\cp$). The \emph{view} of $a \in A$ on history $\sigma$ is defined as:
\[\begin{array}{llllll}
\view_a(\epsilon) & := & \emptyset \hspace{2cm} &
\view_a(\sigma.\cg) & := & \view_{\cg a}(\sigma).\cg a
\end{array}\]
where $\view_{\cg a}(\sigma)$ is the ordered list of views $\view_b(\sigma)$ for $b \in \cg a$. The set of \emph{history variables} is $\SSigma :=$ $\{ (\view_a(\sigma))_a \mid \sigma \text{ a history}, a \in A\}$. Also, $\Sigma^n := \{ (\view_a(\sigma))_a \in \SSigma \mid |\sigma|=n, a \in A\}$, and $\Sigma^{<{n}} = \Union_{m < n} \Sigma^m$.
\end{definition}
The view of agent $a$ on history $\sigma$ defines a {\bf tree} with root $\cg a$ where $\cg$ is the last element of $\sigma$. A history variable for $a$ is nothing but the view of $a$ of a history $\sigma$, subscripted with $a$, denoting a local variable. 
The set of views is known as the \emph{full-information protocol} \cite{MosesT88}. We now model the arbitrary iteration of a communication pattern in an epistemic model, while keeping track of the previous rounds by way of history variables. The definition is for agents $A$ and variables $P \union \SSigma$ (and not, as before, for $A$ and $P$).

\begin{definition}[History epistemic model]
Given an epistemic model $M=(W,\sim,L)$, a communication pattern $\cp$, and $n \in \mathbb N$, a {\em history epistemic model} $M \dot\odot \cp^n$ is defined as follows. For $n=0$, $M \dot\odot \cp^0 = M$. For $n\geq 0$, given $M \dot\odot \cp^n = (W \times \cp^n,\sim,L)$, we define $M \dot\odot \cp^{n+1} := (W \times \cp^{n+1},\sim',L')$\footnote{Allowing slight abuse of the notation $\cp^n$.} such that:
\begin{itemize}
\item $(w,\sigma.\cg) \sim'_a (w', \sigma'.\cg')$ iff $(w,\sigma) \sim_{\cg a} (w',\sigma')$ and $\cg a = \cg' a$;
\item $L'(w,\sigma.\cg) := L(w,\sigma) \union \{(\view_b(\sigma.\cg))_b \mid b \in A \}$.
\end{itemize}
\end{definition}
The domain of $M \dot\odot \cp^n$ is $W \times \cp^n$, so that domain elements have shape $(w,\sigma)$. As $M \dot\odot \cp^0 = M$, all history variables in $M \dot\odot \cp^0$ are false. This means that no round of communication has taken place.

The difference between the $\dot\odot$ update and the $\odot$ update is therefore \emph{only} in the labeling of local variables: we now require a countably infinite set of local history variables such that in each round for each agent the entire history is again recorded by making such a variable true. We will see that this guarantees that interpreted systems are closed under update.

Given the $\dot\odot$ update, the history-based semantics is now as expected, and unlike the previous semantics it has the property that the update of an interpreted system remains an interpreted system.

\begin{definition}[History-based semantics]
Given $M  \dot\odot \cp^n = (W \times \cp^n,\sim,L)$ and $(w,\sigma) \in W$,  define \emph{satisfaction relation} $\models$ by induction on $\phi\in\lang$ (where $p \in P$, $a \in A$, $B \subseteq A$, $\cp$ a communication pattern, $\cg\in\cp$, $\sigma\in\cp^n$, and $\tau \in \cp^*$ --- that is, $\tau$ is an arbitrary history).
\[ \begin{array}{lcl}

M \dot\odot \cp^n, (w,\sigma) \models p_a & \text{iff} & p_a \in L(w)\\

M \dot\odot \cp^n, (w,\sigma) \models (\view_a(\tau))_a & \text{iff} & (\view_a(\tau))_a \in L(w,\sigma)\\
M \dot\odot \cp^n, (w,\sigma)  \models \neg\phi & \text{iff} & M \dot\odot \cp^n, (w,\sigma)  \not\models \phi\\
M \dot\odot \cp^n, (w,\sigma)  \models \phi\et\psi & \text{iff} & M \dot\odot \cp^n, (w,\sigma)  \models \phi \text{ and } M \dot\odot \cp^n, (w,\sigma) \models \psi \\
M \dot\odot \cp^n, (w,\sigma)  \models D_B \phi & \text{iff} & M \dot\odot \cp^n, (v,\tau)  \models \phi \text{ for all } (v,\tau) \sim_B (w,\sigma) \\
M \dot\odot \cp^n, (w,\sigma) \models [\cp,\cg]\phi & \text{iff} & M \dot\odot \cp^{n+1}, (w,\sigma.\cg) \models \phi
\end{array} \]
\end{definition}

\begin{proposition} \label{prop.xxxxis}
Let interpreted system $M$ and $\cp$ be given. Then $M \dot\odot \cp^n$ is an interpreted system. 
\end{proposition}

\begin{proof}
Let $M \dot\odot \cp^n =  (W \times \cp^n,\sim,L)$.
We are required to show that $(w,\sigma) \sim_a (w',\sigma')$ iff $L(w,\sigma)_a = L(w',\sigma')_a$. 

For $n=0$ this is because $M$ is an interpreted system. 

Let us now assume $M \dot\odot \cp^n$ is an interpreted system and consider $M \dot\odot \cp^{n+1}$, and $\cg,\cg' \in \cp$. We then have that (where $|\sigma|=|\sigma'|=n$):

\bigskip

\noindent $
(w,\sigma.\cg) \sim_a (w',\sigma'.\cg') \\
\Eq \hfill \text{by definition of } \sim_a \\
\cg a = \cg' a \text { and } (w,\sigma) \sim_{\cg a} (w',\sigma') \\
\Eq \\
\cg a = \cg' a, \text{ and for all } b \in \cg a: (w,\sigma) \sim_b (w',\sigma') \\
\Eq \hfill \text{inductive hypothesis} \\
\cg a = \cg' a, \text{ and for all } b \in \cg a: L(w,\sigma)_b = L(w',\sigma')_b \\
\Eq \hfill (*) \\
L(w,\sigma.\cg)_a = L(w',\sigma'.\cg')_a
$

\bigskip

$(*)$: By definition, we have that $L(w,\sigma.\cg) = L(w,\sigma) \union \{(\view_b(\sigma.\cg))_b \mid b \in A \}$. Therefore, for agent $a$, we have that $L(w,\sigma.\cg)_a = L(w,\sigma)_a \union \{(\view_a(\sigma.\cg))_a\}$. As $a \in \cg a$, we may assume by induction that $L(w,\sigma)_a = L(w',\sigma')_a$. It therefore remains to show that $(\view_a(\sigma.\cg))_a = (\view_a(\sigma'.\cg'))_a$. By the definition of $\view$, this is equivalent to requiring that $\cg a = \cg' a$, and that $(\view_b(\sigma))_b = (\view_b(\sigma'))_b$ for all $b \in \cg a$. The latter is given above. Concerning the former: from the inductive assumption that $L(w,\sigma)_b = L(w',\sigma')_b$ for all $b \in \cg a$ and the definition of $\view$ for these agents $b$ it follows that $(\view_b(\sigma))_b = (\view_b(\sigma'))_b$ for all $b \in \cg a$.
\end{proof}

In order to compare the update expressivity of action models and communication patterns in this semantics, we must also change Def.~\ref{def.induced} of induced action model $\U(\cp)$. There are now infinitely many local variables, so that the description of a valuation is no longer a formula but an infinite conjunction. However, for every round of communication a description of the valuation of a finite subset is sufficient.
\begin{definition}[Induced action model for round n]
The induced action model $\U^n(\cp) = (E,\sim,\pre)$ for round $n$ of iterated execution of $\cp$ is defined as follows, where $\cg,\cg' \in \cp$, $Q,Q' \subseteq P \union \Sigma^{<{n}}$, and $a \in A$:
\[ \begin{array}{lcl}
E & = & \cp \times \power(P \union \Sigma^{<{n}}) \\
(\cg,Q) \sim_a (\cg',Q') & \text{iff} & \cg a = \cg' a \text{ and } Q_{\cg a} = Q'_{\cg' a} \\
\pre(\cg,Q) & = & \delta_{Q,P \union \Sigma^{<{n}}}
\end{array}\]
\end{definition}
Although $P \union \SSigma$ infinite, $P \union \Sigma^{<{n}}$ is finite. Note that $\U(\cp)$ is $\U^1(\cp)$, where $\delta_Q$ is now $\delta_{Q,P}$, as $\Sigma^{<1}=\Sigma^0=\emptyset$. We recall the definition of $\delta_{Q,P \union \Sigma^{<{n}}}$ from Sect.~\ref{sec.language}.
From Prop.~\ref{prop.xxxxis} and Prop.~\ref{prop.xxxx} we directly obtain:
\begin{proposition} \label{prop.xxxxis2}
Let interpreted system $M$ and communication pattern $\cp$ be given. Then $M \dot\odot \cp^n$ is bisimilar to $M \otimes \U^1(\cp) \otimes \dots \otimes \U^n(\cp)$.
\end{proposition}


By abbreviation inductively define $(\cp^0,\epsilon):=\epsilon$ and $(\cp^{n+1},\sigma.\cg):= (\cp^n,\sigma).(\cp,\cg)$, where $\sigma \in \cp^n$. Recalling the definition of $[\cp]\phi$ as $\Et_{\cg\in\cp} [\cp,\cg]\phi$, we let $[\cp^n]\phi$ stand for $\Et_{\sigma\in\cp^n}[\cp^n,\sigma]\phi$. Just as $[\cp^n]\phi$ is equivalent to $[\cp]^n\phi$, $[\cp,\sigma]\phi$ is equivalent to $[\cp,\cg_1]\dots[\cp,\cg_n]\phi$, where $\sigma=\cg_1\dots\cg_n$.


In the $\odot$ semantics, the answer to the question whether a communication pattern $(\cp,\cg)$ is update equivalent to an action model $(\U(\cp),T)$ where $T=\{(\cg,Q)\mid Q \subseteq P\}$, on the class of epistemic models, was `no' (Example~\ref{example.iis}). This now becomes the question whether in the history-based $\dot\odot$ semantics an iterated communication pattern $(\cp^n,\sigma)$ is update equivalent to a multi-pointed action model on the class of interpreted systems with empty histories. The answer to that is `yes'. However, communication pattern modalities occurring in a formula may not be interpreted in the empty history. For example, given $[\cp,\cg](p_a \imp D_B[\cp,\cg']p_b)$, subformula $[\cp,\cg']p_b$ will be interpreted in some world $(w,\cg)$, not in some world $(w,\epsilon)$. We want it equivalent to some formula of shape $[\cp,\cg\cg']\psi$. We therefore show that any formula $\phi\in\lang^\circ$ is equivalent to one wherein all subformulas $[\cp^n,\sigma]\psi$ have that $\psi\in\lang^-$ (without dynamic modalities). All dynamic modalities are then interpreted in an empty history epistemic model. 

Define the {\em iterated update normal form} (IUNF), the language $\lang^\circ_{\mathsf{iunf}}$ (with members $\phi$) by BNF as:
\[\begin{array}{lll}
\phi & := & p_a \mid \neg \phi \mid \phi \et \phi \mid D_B \phi \mid [\cp^n,\sigma]\psi \\
\psi & := & p_a \mid \neg \psi \mid \psi \et \psi \mid D_B \psi
\end{array}\]
\begin{lemma} \label{lemmab}
Every formula in $\lang^\circ$ is equivalent to one in $\lang^\circ_{\mathsf{iunf}}$, in iterated update normal form.
\end{lemma}

\begin{proof}
We define a translation $t: \lang^\circ \imp \lang^\circ_{\mathsf{iunf}}$. We prove by induction that any $\phi$ is equivalent to $t(\phi)$. All clauses are trivial, and the one for the dynamic modality has a subinduction. The subinduction uses the reduction axioms for communication patterns found in \cite{cdrv:2022}.
\[\begin{array}{lll}
t([\cp,\cg]p_a) & := & [\cp,\cg]p_a  \\
t([\cp,\cg](\phi\et\psi)) & := & t([\cp,\cg]\phi) \et t([\cp,\cg]\psi)  \\
t([\cp,\cg]\neg\phi) & := & \neg t([\cp,\cg]\phi)  \\
t([\cp,\cg] D_B \phi) & := & \Et_{\cg' B \equiv \cg B} D_{\cg B} t([\cp,\cg']\phi)  \\
t([\cp,\cg] [\cp',\cg'] \phi) & := & t([\cp,\cg]t([\cp',\cg'] \phi))
\end{array}\]
In particular, we have that \[\begin{array}{llllllll}
t([\cp,\cg](\phi\vel\psi)) &=& t([\cp,\cg]\neg(\neg\phi\et\neg\psi)) &=& \neg t([\cp,\cg](\neg\phi\et\neg\psi)) &= \\ \neg (t([\cp,\cg]\neg\phi)\et t([\cp,\cg]\neg\psi)) &=& \neg (\neg t([\cp,\cg]\phi)\et \neg t([\cp,\cg]\psi)) &=& t([\cp,\cg]\phi)\vel t([\cp,\cg]\psi) \end{array}\] We recall that notation $\cg' B \equiv \cg B$ was defined in Sect.~\ref{sec.structures}. Further proof details are omitted.
\end{proof}


\begin{proposition} \label{prop.xxxxis3}
Action model logic is at least as update expressive as communication pattern logic on the class of interpreted systems, in the history-based semantics.
\end{proposition}

\begin{proof}
Let an interpreted system $M$ and a communication pattern $\cp$ be given. Then $M \dot\odot \cp^n$ is bisimilar to $M \otimes \U^1(\cp) \otimes \dots \otimes \U^n(\cp)$ (Prop.~\ref{prop.xxxxis2}). Consider the action model $\U$ that is the \emph{composition} of $\U^1(\cp)$, \dots, $\U^n(\cp)$, where we note that, unlike communication patterns, action models are indeed closed under composition (see \cite{baltagetal:1998} for the definition of action model composition).

Let us now consider what action model some $(\cp^n,\sigma)$ is update equivalent to. We can assume that modalities $[\cp^n,\sigma]$ are only interpreted in $M$, a history epistemic model for an empty history (Lemma~\ref{lemmab}). Iterated communication pattern $\cp^n$ is clearly update equivalent to $\U$. It is almost worded as such in Prop.~\ref{prop.xxxxis3}. Also, any $(\cp^n,\sigma)$ is update equivalent to $(\U,T)$, where, if $\sigma = \cg^1\cg^2\dots\cg^n$, \[ T = \{ (\cg^1,Q^1)(\cg^2,Q^2)\dots(\cg^n,Q^n) \mid Q^1 \subseteq P, Q^2 \subseteq \Sigma^1, \dots, Q^n \subseteq\Sigma^{n-1} \}.\] Details are omitted. Note that $P \union \Sigma^1 \union \dots \Sigma^{n-1} = P \union \Sigma^{<n}$, the set of all atoms considered at round $n$.
\end{proof}


It is easy to see that Prop.~\ref{prop.prop} still holds for the history-based semantics. Therefore:
\begin{corollary}
Action model logic is more update expressive than communication pattern logic on the class of interpreted systems, in the history-based semantics.
\end{corollary}


This story on history-based semantics could just as well have been told for sequences $\cp_1,\dots,\cp_n$ of possibly different communication patterns, instead of for $n$ iterations of a given communication pattern $\cp$. We would then get models $M \dot\odot \cp_1 \dot\odot \dots \dot\odot \cp_n$ instead of models $M \dot\odot\cp^n$, and we would get induced action models $M \otimes \U^1(\cp_1) \otimes \dots \otimes \U^n(\cp_n)$, etcetera. However, in distributed computing it is common to consider arbitrary iteration of the same communication pattern (the mentioned oblivious model).

Although in such a generalization we can continue to view histories as sequences of communication graphs, it is important to realize that the same communication graph can then be the point of a different communication pattern, which may give their execution a different meaning. For example, recall $\cg^{ab} b = \{a,b\} \union I$. Given $\cg^{ab} \in \IS$, agent $b$ is uncertain whether $a$ has received his message. But given $\cg^{ab} \in \{\cg^{ab}\}$, the singleton communication pattern, agent $b$ knows that agent $a$ has not received his message.

\section{Conclusions and further research} \label{sect:concl}

We have shown that action model logic and communication pattern logic are incomparable in update expressivity on epistemic models, and that action model logic is more update expressive than communication pattern logic on interpreted systems. It seems promising to investigate communication patterns further, also on epistemic models that are not local (clearly, incomparability does not depend on that). Induced action models are exponentially larger than communication patterns. Communication patterns intuitively specify system dynamics that abstracts from message content. Results in temporal epistemics on synchronous and asynchronous computation should carry over to dynamic epistemics. 

\paragraph*{Acknowledgements} We thank the TARK reviewers for their comments. 
This work was partially supported by Programa de
Apoyo a Proyectos de Investigaci\'on e Innovaci\'on
Tecnol\'ogica (PAPIIT), project grants IN108720 and IN108723. Diego A. Vel\'azquez is the recipient of
a fellowship from CONACyT.

\bibliographystyle{eptcs}
\bibliography{biblio2023}

\providecommand{\noopsort}[1]{}
\begin{thebibliography}{10}
\providecommand{\bibitemdeclare}[2]{}
\providecommand{\surnamestart}{}
\providecommand{\surnameend}{}
\providecommand{\urlprefix}{Available at }
\providecommand{\url}[1]{\texttt{#1}}
\providecommand{\href}[2]{\texttt{#2}}
\providecommand{\urlalt}[2]{\href{#1}{#2}}
\providecommand{\doi}[1]{doi:\urlalt{https://doi.org/#1}{#1}}
\providecommand{\eprint}[1]{arXiv:\urlalt{https://arxiv.org/abs/#1}{#1}}
\providecommand{\bibinfo}[2]{#2}

\bibitemdeclare{article}{AgotnesW17}
\bibitem{AgotnesW17}
\bibinfo{author}{T.~\surnamestart {\AA}gotnes\surnameend} \&
  \bibinfo{author}{Y.N. \surnamestart W{\'{a}}ng\surnameend}
  (\bibinfo{year}{2017}): \emph{\bibinfo{title}{Resolving distributed
  knowledge}}.
\newblock {\slshape \bibinfo{journal}{Artif. Intell.}} \bibinfo{volume}{252},
  pp. \bibinfo{pages}{1--21}, \doi{10.1016/j.artint.2017.07.002}.

\bibitemdeclare{inproceedings}{baltagetal:1998}
\bibitem{baltagetal:1998}
\bibinfo{author}{A.~\surnamestart Baltag\surnameend}, \bibinfo{author}{L.S.
  \surnamestart Moss\surnameend} \& \bibinfo{author}{S.~\surnamestart
  Solecki\surnameend} (\bibinfo{year}{1998}): \emph{\bibinfo{title}{The Logic
  of Public Announcements, Common Knowledge, and Private Suspicions}}.
\newblock In: {\slshape \bibinfo{booktitle}{Proc.\ of 7th TARK}}, pp.
  \bibinfo{pages}{43--56}.

\bibitemdeclare{inproceedings}{Baltag20}
\bibitem{Baltag20}
\bibinfo{author}{A.~\surnamestart Baltag\surnameend} \&
  \bibinfo{author}{S.~\surnamestart Smets\surnameend} (\bibinfo{year}{2020}):
  \emph{\bibinfo{title}{Learning What Others Know}}.
\newblock In: {\slshape \bibinfo{booktitle}{Proc.\ of 23rd {LPAR}}}, {\slshape
  \bibinfo{series}{EPiC Series in Computing}}~\bibinfo{volume}{73}, pp.
  \bibinfo{pages}{90--119}, \doi{10.29007/plm4}.

\bibitemdeclare{article}{jfaketal.lcc:2006}
\bibitem{jfaketal.lcc:2006}
\bibinfo{author}{J.~\surnamestart van Benthem\surnameend},
  \bibinfo{author}{J.~\surnamestart van Eijck\surnameend} \&
  \bibinfo{author}{B.~\surnamestart Kooi\surnameend} (\bibinfo{year}{2006}):
  \emph{\bibinfo{title}{Logics of Communication and Change}}.
\newblock {\slshape \bibinfo{journal}{Information and Computation}}
  \bibinfo{volume}{204(11)}, pp. \bibinfo{pages}{1620--1662},
  \doi{10.1016/j.ic.2006.04.006}.

\bibitemdeclare{book}{blackburnetal:2001}
\bibitem{blackburnetal:2001}
\bibinfo{author}{P.~\surnamestart Blackburn\surnameend},
  \bibinfo{author}{M.~\surnamestart de~Rijke\surnameend} \&
  \bibinfo{author}{Y.~\surnamestart Venema\surnameend} (\bibinfo{year}{2001}):
  \emph{\bibinfo{title}{Modal Logic}}.
\newblock \bibinfo{publisher}{Cambridge University Press},
  \doi{10.1017/CBO9781107050884}.

\bibitemdeclare{article}{cdrv:2022}
\bibitem{cdrv:2022}
\bibinfo{author}{A.~\surnamestart Casta{\~{n}}eda\surnameend},
  \bibinfo{author}{H.~\surnamestart van Ditmarsch\surnameend},
  \bibinfo{author}{D.A. \surnamestart Rosenblueth\surnameend} \&
  \bibinfo{author}{D.A. \surnamestart Vel\'azquez\surnameend}
  (\bibinfo{year}{2022}): \emph{\bibinfo{title}{Communication Pattern Logic:
  Epistemic and Topological Views}}.
\newblock {\slshape \bibinfo{journal}{CoRR}} \bibinfo{volume}{abs/2207.00823},
  \doi{10.48550/arXiv.2207.00823}.
\newblock \bibinfo{note}{To appear in {\em Journal of Philosophical Logic}}.

\bibitemdeclare{article}{hvdetal.aus:2020}
\bibitem{hvdetal.aus:2020}
\bibinfo{author}{H.~\surnamestart van Ditmarsch\surnameend},
  \bibinfo{author}{W.~\surnamestart van~der Hoek\surnameend},
  \bibinfo{author}{B.~\surnamestart Kooi\surnameend} \& \bibinfo{author}{L.B.
  \surnamestart Kuijer\surnameend} (\bibinfo{year}{2020}):
  \emph{\bibinfo{title}{Arrow Update Synthesis}}.
\newblock {\slshape \bibinfo{journal}{Information and Computation}}, p.
  \bibinfo{pages}{104544}, \doi{10.1016/j.ic.2020.104544}.

\bibitemdeclare{article}{hvdetal.FAMAS:2013}
\bibitem{hvdetal.FAMAS:2013}
\bibinfo{author}{H.~\surnamestart van Ditmarsch\surnameend},
  \bibinfo{author}{J.\ \surnamestart Ruan\surnameend} \& \bibinfo{author}{W.\
  \surnamestart van~der Hoek\surnameend} (\bibinfo{year}{2013}):
  \emph{\bibinfo{title}{Connecting Dynamic Epistemic and Temporal Epistemic
  Logics}}.
\newblock {\slshape \bibinfo{journal}{Logic journal of the IGPL}}
  \bibinfo{volume}{21(3)}, pp. \bibinfo{pages}{380--403},
  \doi{10.1093/jigpal/jzr038}.

\bibitemdeclare{article}{DworkM90}
\bibitem{DworkM90}
\bibinfo{author}{C.~\surnamestart Dwork\surnameend} \&
  \bibinfo{author}{Y.~\surnamestart Moses\surnameend} (\bibinfo{year}{1990}):
  \emph{\bibinfo{title}{Knowledge and Common Knowledge in a {B}yzantine
  Environment: Crash Failures}}.
\newblock {\slshape \bibinfo{journal}{Inf. Comput.}}
  \bibinfo{volume}{88}(\bibinfo{number}{2}), pp. \bibinfo{pages}{156--186},
  \doi{10.1016/0890-5401(90)90014-9}.

\bibitemdeclare{article}{jveetal:2012}
\bibitem{jveetal:2012}
\bibinfo{author}{J.~\surnamestart van Eijck\surnameend},
  \bibinfo{author}{J.~\surnamestart Ruan\surnameend} \&
  \bibinfo{author}{T.~\surnamestart Sadzik\surnameend} (\bibinfo{year}{2012}):
  \emph{\bibinfo{title}{Action emulation}}.
\newblock {\slshape \bibinfo{journal}{Synthese}} \bibinfo{volume}{185(1)}, pp.
  \bibinfo{pages}{131--151}, \doi{10.1007/s11229-012-0083-1}.

\bibitemdeclare{book}{herlihyetal:2013}
\bibitem{herlihyetal:2013}
\bibinfo{author}{M.~\surnamestart Herlihy\surnameend},
  \bibinfo{author}{D.~\surnamestart Kozlov\surnameend} \&
  \bibinfo{author}{S.~\surnamestart Rajsbaum\surnameend}
  (\bibinfo{year}{2013}): \emph{\bibinfo{title}{Distributed Computing Through
  Combinatorial Topology}}.
\newblock \bibinfo{publisher}{Morgan Kaufmann}, \doi{10.1016/C2011-0-07032-1}.

\bibitemdeclare{inproceedings}{kooirenne}
\bibitem{kooirenne}
\bibinfo{author}{B.~\surnamestart Kooi\surnameend} \&
  \bibinfo{author}{B.~\surnamestart Renne\surnameend} (\bibinfo{year}{2011}):
  \emph{\bibinfo{title}{Generalized Arrow Update Logic}}.
\newblock In: {\slshape \bibinfo{booktitle}{Proc.\ of 13th {TARK}}}, pp.
  \bibinfo{pages}{205--211}, \doi{10.1145/2000378.2000403}.

\bibitemdeclare{article}{lamportetal:1982}
\bibitem{lamportetal:1982}
\bibinfo{author}{L.~\surnamestart Lamport\surnameend},
  \bibinfo{author}{R.~\surnamestart Shostak\surnameend} \&
  \bibinfo{author}{M.~\surnamestart Pease\surnameend} (\bibinfo{year}{1982}):
  \emph{\bibinfo{title}{The {B}yzantine Generals Problem}}.
\newblock {\slshape \bibinfo{journal}{ACM Trans. Program. Lang. Syst.}}
  \bibinfo{volume}{4}(\bibinfo{number}{3}), pp. \bibinfo{pages}{382--401},
  \doi{10.1145/357172.357176}.

\bibitemdeclare{article}{MosesT88}
\bibitem{MosesT88}
\bibinfo{author}{Y.~\surnamestart Moses\surnameend} \& \bibinfo{author}{M.R.
  \surnamestart Tuttle\surnameend} (\bibinfo{year}{1988}):
  \emph{\bibinfo{title}{Programming Simultaneous Actions Using Common
  Knowledge}}.
\newblock {\slshape \bibinfo{journal}{Algorithmica}} \bibinfo{volume}{3}, pp.
  \bibinfo{pages}{121--169}, \doi{10.1007/BF01762112}.

\bibitemdeclare{inproceedings}{plaza:1989}
\bibitem{plaza:1989}
\bibinfo{author}{J.A. \surnamestart Plaza\surnameend} (\bibinfo{year}{1989}):
  \emph{\bibinfo{title}{Logics of Public Communications}}.
\newblock In: {\slshape \bibinfo{booktitle}{Proc.\ of the 4th ISMIS}},
  \bibinfo{publisher}{Oak Ridge National Laboratory}, pp.
  \bibinfo{pages}{201--216}.

\bibitemdeclare{article}{Roelofsen07}
\bibitem{Roelofsen07}
\bibinfo{author}{F.~\surnamestart Roelofsen\surnameend} (\bibinfo{year}{2007}):
  \emph{\bibinfo{title}{Distributed knowledge}}.
\newblock {\slshape \bibinfo{journal}{Journal of Applied Non-Classical Logics}}
  \bibinfo{volume}{17}(\bibinfo{number}{2}), pp. \bibinfo{pages}{255--273},
  \doi{10.3166/jancl.17.255-273}.

\bibitemdeclare{inproceedings}{diego:2021}
\bibitem{diego:2021}
\bibinfo{author}{D.A. \surnamestart Vel\'azquez\surnameend},
  \bibinfo{author}{A.~\surnamestart Casta\~{n}eda\surnameend} \&
  \bibinfo{author}{D.A. \surnamestart Rosenblueth\surnameend}
  (\bibinfo{year}{2021}): \emph{\bibinfo{title}{Communication Pattern Models:
  an Extension of Action Models for Dynamic-Network Distributed Systems}}.
\newblock In: {\slshape \bibinfo{booktitle}{Proc.\ of TARK {XVIII}}}, {\slshape
  \bibinfo{series}{{EPTCS}}} \bibinfo{volume}{335}, pp.
  \bibinfo{pages}{307--321}, \doi{10.4204/EPTCS.335.29}.

\bibitemdeclare{article}{WangA15}
\bibitem{WangA15}
\bibinfo{author}{Y.N. \surnamestart W{\'{a}}ng\surnameend} \&
  \bibinfo{author}{T.~\surnamestart {\AA}gotnes\surnameend}
  (\bibinfo{year}{2015}): \emph{\bibinfo{title}{Relativized common knowledge
  for dynamic epistemic logic}}.
\newblock {\slshape \bibinfo{journal}{J. Appl. Log.}}
  \bibinfo{volume}{13}(\bibinfo{number}{3}), pp. \bibinfo{pages}{370--393},
  \doi{10.1016/j.jal.2015.06.004}.

\end{thebibliography}

\end{document}